\documentclass[11pt]{article} 
\usepackage{amssymb,fullpage}
\usepackage{lineno}
\usepackage[utf8]{inputenc}
\usepackage[english]{babel}
\usepackage{amssymb,amsmath,amsthm, amsfonts}
\usepackage{epsfig}
\usepackage{subcaption}
\newcommand{\ems}[1]{{#1}}

\newtheorem{theorem}{Theorem}
\newtheorem{lemma}{Lemma}
\newtheorem{claim}{Claim}

\newenvironment{claimProof}[1]{\par\noindent\underline{Proof:}\space#1}{\hfill $\blacksquare$}

\newcommand{\remove}[1]{}
\newcommand{\Section}[1]{\section{#1}}
\newcommand{\Subsection}[1]{\subsection{#1}}
\newcommand{\B}{}
\newcommand{\BB}{}
\newcommand{\BBB}{}

\usepackage{framed}
\usepackage[boxed,ruled,vlined,linesnumbered]{algorithm2e}
\newcommand{\BigO}[0]{{\cal O}\xspace}

\begin{document}


\title{Self-Stabilizing \ems{Automatic Repeat Request Algorithms for (Bounded Capacity, Omitting, Duplicating and non-FIFO) Computer Networks}}

\author{Shlomi Dolev~\footnote{Department of Computer Science, Ben-Gurion University of the Negev, Beer-Sheva, Israel \texttt{\{dolev, hanemann, sharmas\}@cs.bgu.ac.il}} \and Ariel Hanemann $^\ast$ \and Elad Michael Schiller \footnote{Department of Computer Science and Engineering, Chalmers University of Technology, G\"{o}teborg, Sweden, \texttt{elad@chalmers.se}} \and Shantanu Sharma $^\ast$}


%
%
%
%

\maketitle

\begin{abstract}
End-to-end communication over the network layer (or data link in overlay networks) is one of the most important communication tasks in every communication network, including legacy communication networks as well as mobile ad hoc networks, peer-to-peer networks and mesh networks. Reliable end-to-end communications are based on Automatic Repeat reQuest (ARQ) algorithms for dealing with packet failures, such as packet drops. We study ARQ algorithms that exchange packets to deliver (high level) messages in first-in-first-out (FIFO) order without omissions or duplications. We present a self-stabilizing ARQ algorithm that can be applied to networks of bounded capacity that \ems{are prone to packet loss, duplication, and reordering.} 
%
%
\ems{Our analysis considers Lamport's happened-before relation when demonstrating stabilization without assuming the presence of a fair scheduler. It shows that the length of the longest chain of Lamport's happened-before relation is $8$ for any system run.} 
\end{abstract}




\section{Introduction}



\ems{Reliable} end-to-end communication is a basic primitive in communication networks. A {\em Sender} must deliver to a {\em Receiver} one message at a time, where no omissions, duplications and reordering are allowed. \ems{Failures may} occur in transmitting packets among the network entities -- one significant \ems{cause of failure} is noise in the transmission media. Thus, error detection and error correction methods are employed as an integral part of the transmission in the communication network. These error detection and correction codes function with high probability. Still, when there is a large volume of communication sessions, the probability that an error will not be detected becomes high, leading to a possible malfunction of the communication algorithm. In fact, it can lead the algorithm to an arbitrary state from which the algorithm may never recover unless it is \emph{self-stabilizing}~(\cite{D2K}, Chapter~$3$). By using packets with enough distinct labels infinitely often, we \ems{propose the design of  a self-stabilizing Automatic Repeat reQuest (ARQ)} algorithm that can be applied to computer networks of bounded capacity that \ems{may} omit, duplicate, and reorder packets.


In the context of existing large-scale networks, such as the Internet, it is often the case that the correctness of transport layer protocols, such as TCP/IP, depends on the correct estimation of the maximum segment lifetime (MSL). This bound on the time that a TCP segment can exist in the Internet is arbitrarily defined to be $2$ minutes long.~\footnote{``RFC 793''. Transmission Control Protocol. September 1981.} In some systems, MSL determines the time in which the kernel waits before releasing the resources that are bound to a partially closed TCP/IP connection~\cite{stevens2004unix}. Thus, failing to estimate MSL correctly, say, too quickly, can cause the delivery of TCP segments that were not sent during the lifetime of the current connection (they were rather sent during previous connection incarnations). When selecting an MSL value that is too long, the host becomes exposed to attacks in which a selective packet omission of the connection closing acknowledgment can delay the resource release until the MSL period is over. We generalize the above network anomalies and consider asynchronous settings in which, in addition to message omission and duplication, the Receiver can receive messages from the network, \ems{i.e., messages} that the Sender have not transmitted (during the connection lifetime since these messages are rather the result of transient faults). We bound by the network capacity the number of \ems{such messages, which} the Sender have not transmitted (during the connection lifetime).


The dynamic and difficult-to-predict nature of large scale networks gives rise to many fault-tolerance issues that are hard to efficiently manage and control. One would prefer a system that automatically recovers from unexpected failures, possibly as part of after-disaster recovery, transient faults due to hardware or software temporal malfunctions or even after benign temporal violation of the assumptions made in the system design. 

For example, the assumption that error detection ensures the arrival of correct messages and the discarding of corrupted messages. In practice, error detection is a probabilistic mechanism that may not detect a corrupted message, and therefore, the message can be regarded legitimate, driving the system to an arbitrary state after which, availability and functionality may be damaged forever, unless there is human intervention.

Fault-tolerant systems that are {\em self-stabilizing}~\cite{D2K,DBLP:journals/cacm/Dijkstra74} can recover after the occurrence of \ems{arbitrary} transient faults, which can drive the system to an arbitrary system state. The \ems{designers of self-stabilizing systems consider all system states as a possible starting state.} There is a rich research literature about Automatic Repeat reQuest (ARQ) techniques for obtaining fault-tolerant protocol that provide end-to-end message delivery. However, when initiating a system in an arbitrary state, \ems{all non-self-stabilizing algorithms cannot provide any} guarantee that the system will reach a legal state after which the participants maintain a coherent state. \ems{Thus, the} self-stabilization design criteria liberate the \ems{application} designer from dealing with specific fault scenarios, the risk of neglecting some scenarios, and having to address each fault scenario separately.

New challenges appear when designing self-stabilizing \ems{ARQ algorithms.} One significant challenge is to provide an ordering for message transmitted between the Sender and the Receiver, which is an even more intriguing problem when the system starts from an arbitrary state. Usually, new messages are identified by a new message number; a number greater than all previously used numbers. Counters of $64$-bits, or so, are usually used to implement such numbers. Such designs \ems{are often} justified by claiming that $64$-bit values suffice for implementing (practically) unbounded counters. However, a single (arbitrary) transient fault can cause the counter to reach the upper limit at once. \ems{Thus, all counters (as well as sequence numbers and labels) that our self-stabilizing algorithm uses have predefined sizes.}

\Subsection{Task description}
\ems{We follow the task description that appears in earlier self-stabilizing \ems{Automatic Repeat reQuest (ARQ)} algorithms~\cite{D2K,DBLP:journals/ipl/DolevDPT11}.}   
We consider a pair that includes a Sender and a Receiver.
The Sender fetches (one at a time) messages from an ordered sequence of its application-layer input, and the Receiver delivers these messages to its application layer.
To that end, the Sender and the Receiver are connected through unreliable communication channels that allow them to exchange packets. 
We focus on bounded capacity communication channels that can adversely duplicate packets, have them received not in the order in which they were sent (i.e., non-FIFO networks) and omit packets (as long as when one of the communication end transmits a packet infinitely often, the other end receives them  infinitely often). 

Our view on the \ems{Automatic Repeat reQuest (ARQ)} task considers the design of a self-stabilizing \ems{ARQ} algorithms that can (\textit{i}) ensure exactly one copy of message delivery and in the same order as these messages were fetched from the application layer, (\textit{ii}) handling (packet-level) corruption, omission, and duplication of by the communication channel, (\textit{iii}) ensuring applicability to computer networks, and (\textit{iv}) recovering autonomously after the \ems{(occurrence of the} last) transient fault, which may leave the system in an arbitrary starting state for \ems{the communication endpoints as well as} the communication channels between them (including the relay nodes).

\Subsection{Related work}
End-to-end communication and data-link algorithms are fundamental for any network protocol~\cite{DBLP:books/daglib/0008392}.  \ems{Automatic Repeat reQuest (ARQ)} algorithms provide the means for message exchange between senders and receivers over unreliable communication links. Not all end-to-end communications and data-link algorithms assume initial synchronization between senders and receivers. 
Alternating bit protocol (ABP) can transmit data over an unreliable channel~\cite{Tel2001}. It is a special case of go-back-N and sliding window protocols~\cite{DBLP:books/daglib/0008392} when the window size is one. In ABP, the Sender and the Receiver use only an index with two stated, say 0 and 1, that can be encoded by a single bit. ABP does not consider non-FIFO channels, as we do. The proposed algorithm considers an index with three states in order to overcome the challenge of non-FIFO channels. Details about sliding window protocol may be found in~\cite{DBLP:books/daglib/0008392,forouzan2006data}.
Afek and Brown~\cite{DBLP:journals/dc/AfekB93} present a self-stabilizing alternating bit protocol (ABP) for first-in-first-out (FIFO) packet channels without the need for initial synchronization. A method for self-stabilizing token passing was used as the basis for self-stabilizing ABP over unbounded capacity and FIFO preserving channels in~\cite{DBLP:journals/tc/GoudaM91,DBLP:journals/siamcomp/DolevIM97}. Spinelli~\cite{DBLP:journals/ton/Spinelli97} introduced two self-stabilizing sliding window ARQ protocols for unbounded FIFO channels. Dolev and Welch~\cite{DBLP:journals/tc/DolevW97} consider the bare network to be a network with FIFO non-duplicating communication links, and use source routing over paths to cope with crashes. 
%
%
We instead of looking into all the system details, network topology and routing policy, we merely consider the total network capacity and assume that it is bounded.
%

The authors of~\cite{DBLP:conf/focs/AwerbuchPV91,Varghese93} present a self-stabilizing unit capacity data link algorithm over omitting, and non-FIFO, yet non-duplicating channel. Flauzac and Villain~\cite{DBLP:conf/ispan/FlauzacV97} describe a snapshot algorithm that uses bidirectional and FIFO communication channels. Cournier et al.~\cite{DBLP:conf/ipps/CournierDV09} consider an algorithm for message forwarding over message switched network. This algorithm is snap-stabilizing \cite{DBLP:conf/wss/BuiDPV99} and it ensures one-time-delivery of the emitted message to the destination within a finite time using destination based buffer graph. \ems{Cournier et al., however, does not tolerate duplication, as we do since they assume underline FIFO packet delivery. Moreover, our algorithm design aims at arbitrary system runs for which there are no guarantees regarding scheduling fairness. To that end, our analysis studies the length of the longest chains of Lamport's happened-before relation during period of system recovery.}

%
In the context of computer networks, a non-self-stabilizing protocol, called \textit{Slide}, is given in~\cite{DBLP:conf/podc/AfekGR92}, which does not consider packet loss. Following the \textit{Slide} protocol, Kushilevitz et al.~\cite{DBLP:conf/stoc/KushilevitzOR95} reduced the space that each mobile node needs by adding a message cancelling policy. However, this non-self-stabilizing protocol does not cope with packet duplication. 


In~\cite{DBLP:conf/pdcat/BeinMY09}, a self-stabilizing transformer is given that emulates the reliable channel by embedding the virtual topology on the real topology. Dela$\ddot{e}$t et al.~\cite{DBLP:journals/jpdc/DelaetDNT10} provided a snap-stabilizing propagation of information with feedback (PIF) algorithm for unit capacity channel that is capable to handle duplication by the Sender and loss of packets. Cournier et al.~\cite{DBLP:conf/sss/CournierDLPV10} provides a snap stabilizing message forwarding protocol over the linear chain of nodes. This message forwarding algorithm uses four buffers per link. However, all the protocols that we mentioned above do not cater the duplication by the channel.

Recently, Dolev et al.~\cite{DBLP:journals/ipl/DolevDPT11} presented a self-stabilizing data link algorithm for reliable FIFO message delivery over bounded non-FIFO and non-duplicating channel. We do consider duplicating channels. Moreover, the algorithm in~\cite{DBLP:journals/ipl/DolevDPT11} \ems{overflows} the channel by sending an identical message repeatedly until the Sender collects a sufficient number of acknowledgments. Furthermore, for every message that the sender fetches, the algorithm at the Sender and the Receiver use explicitly synchronization, unlike the solution that we propose here. Dolev et al.~\cite{DBLP:journals/ipl/DolevDPT11} demonstrate that their algorithm is fault-resilience optimal with respect to FIFO-preserving communication channels. We note that their algorithm is not fault-resilience optimal with respect to non-FIFO preserving communication channels that we consider here.
\ems{An earlier version of this work appeared} as an extended abstract in~\cite{DBLP:conf/sss/DolevHSS12}.

\Subsection{Our contribution}
We investigate the basic networking tasks of \ems{data-link protocols as well as reliable end-to-end communications over the network layer or overlay networks.}
%
\ems{Towards facilitating the design of these fundamental protocols, this} paper presents the first, to the best of our knowledge, self-stabilizing \ems{Automatic Repeat reQuest (ARQ)} algorithms for reliable FIFO message delivery over bounded non-FIFO and duplicating \ems{communication channels.} We provide a rigorous correctness proof and demonstrate the  self-stabilizing \ems{closure and convergence properties. Our analysis considers Lamport's happened-before relation when demonstrating stabilization without assuming the presence of a fair scheduler. It shows that the length of the longest chain of Lamport's happened-before relation is $8$ for any system run.}

\Section{System Settings}
\label{s:sys}
Node-to-node protocols at the data link layers, as well as end-to-end protocols at the transport layer, use Automatic Repeat reQuest (ARQ) algorithms for facilitating reliable data communication protocols over unreliable media of communication. We describe our assumptions about the system and network.




\Subsection{Communication channels}
The system establishes bidirectional communication between the Sender, $p_s$, and the Receiver, $p_r$, which may not be connected directly. Namely, between $p_s$ and $p_r$ there is a unidirectional {\em (communication) channel} (modeled as a packet set {and denoted by $channel_{s,r}$}) that transfers packets from $p_s$ to $p_r$, and another unidirectional channel that transfers packets from $p_r$ to $p_s$. The Sender and the Receiver can be part of a network that its topology is depicted by a {\em (communication) graph} of $N$ {\em nodes} (or processors), $p_1$, $p_2$, $\ldots$, $p_N$, \ems{where $N$ is unknown the end-to-end peers.} The graph has {\em (direct communication) links}, $(p_i, p_j)$, whenever $p_i$ can directly send packets to its {\em neighbor}, $p_j$ (without the use of network layer protocols). When node $p_i$ sends a packet, $pckt$, to node $p_j$, the operation $send$ adds a copy of $pckt$ to $channel_{i,j}$. We intentionally do not specify (the possibly unreliable) underlying mechanisms that are used to forward a packet from $p_i$ to $p_j$, e.g., flood routing and shortest path routing, as well as packet forwarding protocols. Once $pckt$ arrives at $p_j$, \ems{it} triggers the $receive$ event, and deletes $pckt$ from the channel set. 

\Subsection{The interleaving model}
Every node, $p_i$, executes a program that is a sequence of {\em (atomic) steps}, where a step starts with local computations and ends with a communication operation, which is either $send$ or $receive$ of a packet. For ease of description, we assume the interleaving model, where steps are executed atomically; a single step at any given time. An input event can either be a packet reception or a periodic timer going off triggering $p_i$ to send. Note that the system is asynchronous (while assuming fair communication but not execution fairness~\cite{D2K}). The non-fixed spontaneous node actions and node processing rates are irrelevant to the correctness proof.

The {\em state}, $s_i$, of a node $p_i$ consists of the value of all the variables of the node including the set of all incoming communication channels. The execution of an algorithm step can change the node's state, and the communication channels that are associated with it. The term {\em (system) configuration} is used for a tuple of the form $(s_1, s_2, \cdots, s_N)$, where each $s_i$ is the state of node $p_i$ (including packets in transit for $p_i$). We define an {\em execution (or run)} $R={c_0,a_0,c_1,a_1,\ldots}$ as an alternating sequence of system configurations, $c_x$, and steps $a_x$, such that each configuration $c_{x+1}$ (except the initial configuration $c_0$) is obtained from the preceding configuration, $c_x$, by the execution of the steps $a_x$. We often associate the step index notation with its executing node $p_i$ using a second subscript, i.e., $a_{i_x}$. We represent the omissions, duplications, and reordering using environment steps that are interleaved with the steps of the processors in $R$.


\Subsection{The task}
We define the system's task by a set of executions called {\em legal executions} ($LE$) in which the task's requirements hold. A configuration $c$ is a {\em safe configuration} for an algorithm and the task of $LE$ provided that any execution that starts in $c$ is a legal execution, which belongs to $LE$. The proposed {\em self-stabilizing  \ems{Automatic Repeat reQuest (ARQ)} communication} algorithm satisfies the $S^2ARQ$ task and provides FIFO and exactly once-delivery guarantees.

%
%
In detail, given a system execution, $R$, and a pair, $p_s$ and $p_r$, of sending and receiving nodes, we associate the message sequence $s_R= {{im}_{0}}, {{im}_{1}}, {{im}_{2}}, \ldots$, which are fetched by $p_s$, with the message sequence $r_R= {{om}}_{0}, {{om}_{1}}, {{om}_{2}}, \ldots$, which are delivered by $p_r$. Note that we list messages according to the order they are fetched (from the higher level application) by the Sender, thus two or more (consecutive or non-consecutive) messages can be identical. The $S^2ARQ$ task requires that for every legal execution, $R \in LE_{}$, there is an infinite suffix, $R^\prime$, in which infinitely many messages are delivered, and $s_{R^\prime} = r_{R^\prime}$. 





\ems{\subsection{The Fault Model}
We model a fault occurrence as a step that the environment takes rather than the algorithm. The studied communication environment is unreliable in the sense that packets are not actually received by their sending order, as we specify next.} 

\subsubsection{\ems{Bounded} channel capacity}
We assume that, at any given time, the entire number of packets in the system does not exceed a known bound, which we call $capacity$. This bound can be calculated by considering the possible number of network links, number of system nodes, the (minimum and maximum) packet size and the amount of memory that each node allocates for each link. \ems{We note that, in the context of self-stabilization, bounded channel capacity is a prerequisite for achieving the studied task (\cite{D2K}, Chapter 3). We clarify that when node $p_i$ sends a packet, $pckt$, to node $p_j$, the operation $send$ adds a copy of $pckt$ to $channel_{i,j}$, as long as the system follows the assumption about the upper bound on the number of packets in the channel. The environment can do that by (i) omitting from $channel_{i,j}$ any packet, or (ii) simply ignoring this send operation, i.e., omitting $pckt$.} 


\ems{\subsubsection{Communication faults}}
We consider solutions that are oriented towards asynchronous message-passing systems. Thus, they are oblivious to the time in which the packets arrive and depart. We assume that, at any time, the communication channels are prone to packet faults, such as loss (omission), duplication, reordering, as long as the system does not violate the channel capacity bound. 

\ems{\subsubsection{Communication fairness}
We consider solutions that are oriented towards asynchronous message-passing systems. Thus, they are oblivious to the time in which the packets arrive and depart. We assume that, at any time, the communication channels are prone to packet faults, such as loss (omission), duplication, reordering, as long as the system does not violate the channel capacity bound. Moreover, we assume that if $p_i$ sends a packet infinitely often to $p_j$, node $p_j$ receives that message infinitely often. We refer to the latter as the \emph{fair communication} assumption.} Note that communication fairness does not imply execution fairness~\cite{D2K}. Moreover, if the communication channels between $p_s$ and $p_r$ are not fair, then the adversary has the power to permanently stop effective communication between the \ems{communicating peers. Furthermore, we note that neither the assumption about the bounded channel capacity nor the communication fairness assumption can prevent from the environment from injecting any of the above failures infinitively often.}      

\ems{\subsubsection{Arbitrary transient faults}
We consider any violation of the assumptions according to which the system was designed to operate. We refer to these violations and deviations as \emph{arbitrary transient faults} and assume that they can corrupt the system state arbitrarily (while keeping the program code intact). The occurrence of an arbitrary transient fault is rare. Thus, as in~\cite{D2K,DBLP:journals/cacm/Dijkstra74}, our model assumes that the last arbitrary transient fault occurs before the system execution starts. Moreover, it leaves the system to start in an arbitrary state.} Note that \ems{arbitrary} transient faults can bring the system to consist of arbitrary, and yet capacity bounded, channel sets.

\ems{\subsection{Self-stabilization}}

\ems{The self-stabilization design criterion was introduced by Dijkstra~\cite{DBLP:journals/cacm/Dijkstra74}.}

\ems{\subsubsection{Dijkstra's self-stabilization criterion}
\label{sec:Dijkstra}
An algorithm is \emph{self-stabilizing} with respect to the task of $LE$, when every (unbounded) execution $R$ of the algorithm reaches within a finite period a suffix $R_{legal} \in LE$ that is legal. That is, Dijkstra~\cite{DBLP:journals/cacm/Dijkstra74} requires that $\forall R:\exists R': R=R' \circ R_{legal} \land R_{legal} \in LE \land |R'| \in \mathbb{N}$, where the operator $\circ$ denotes that $R=R' \circ R''$ concatenates $R'$ with $R''$.}

\ems{\subsubsection{Complexity Measures and Lamport's happened-before relation}
\label{sec:timeComplexity}
The main complexity measure of self-stabilizing algorithms, called \emph{stabilization time} (or recovery time), refers to period that it takes the system to recover after the occurrence of the last transient fault. Lamport~\cite{DBLP:journals/cacm/Lamport78} defined the happened-before relation as the least strict partial order on events for which: (i) If steps $a, b \in R$ are taken by processor $p_i$, the happened-before relation $a \rightarrow b$ holds if $a$ appears in $R$ before $b$. (ii) If step $a$ includes sending a message $m$ that step $b$ receives, then $a \rightarrow b$. Using the happened-before definition, one can create a directed acyclic (possibly infinite) graph $G_R:(V_R,E_R)$, where the set of nodes, $V_R$, represents the set of system states in $R$. Moreover, the set of edges, $E_R$, is given by the happened-before relation. In this paper, we assume that the weight of an edge that is due to cases (i) and (ii) are zero and one, respectively. Since we do not assume any guarantee that execution $R$ is fair, we consider the weight of the heaviest directed path between two system state $c,c' \in R$ as the cost measure between $c$ and $c'$.}      

%

\Section{Background and Basic Results}
\label{s:bck}
Network protocols uses a variety of techniques for increasing their robustness, such as routing over  many paths and retransmissions. These techniques can cause erroneous behavior, e.g., message duplications and reordering. For the presentation's simplicity sake, we start, as a first attempt, with an ARQ algorithm that is self-stabilizing and copes network faults, such as packet omissions, duplications, and reordering. This first attempt algorithm has a large \ems{cost}, but it prepares the presentation of our proposal for an efficient solution (Section~\ref{s:alg}) that is based on error correction codes.

\Subsection{A first attempt solution}
We regard two nodes, $p_s$ and $p_r$, as sender, and respectively, receiver; see our first attempt sketch of an \ems{ARQ algorithm} in Figure~\ref{fig:smplalg}. The goal is for $p_s$ to fetch messages, $m$, from its application layer, send $m$ over the communication channel, and for $p_r$ to deliver $m$ to its application layer exactly once and in the same order by which the Sender fetched them from its application layer. The Sender, $p_s$, fetches the message $m$ and starts the transmission of $(2 \cdot capacity+1)$ copies of $m$ to $p_r$, and $p_r$ acknowledges $m$ upon arrival. These transmissions use distinct labels for each copy, i.e., $(2 \cdot capacity+1)$ labels for each of $m$'s copies. The Sender, $p_s$, does not stop retransmitting $m$'s packets until it receives from $p_r$  $(capacity+1)$ distinctly labeled acknowledgment packets, see details in Figure~\ref{fig:smplalg}.


\begin{figure*}
	\BBB\BBB\BBB\BBB\BBB
	\begin{smaller}

		
		
		\begin{framed}
			\begin{center}
				\includegraphics[scale=0.8]{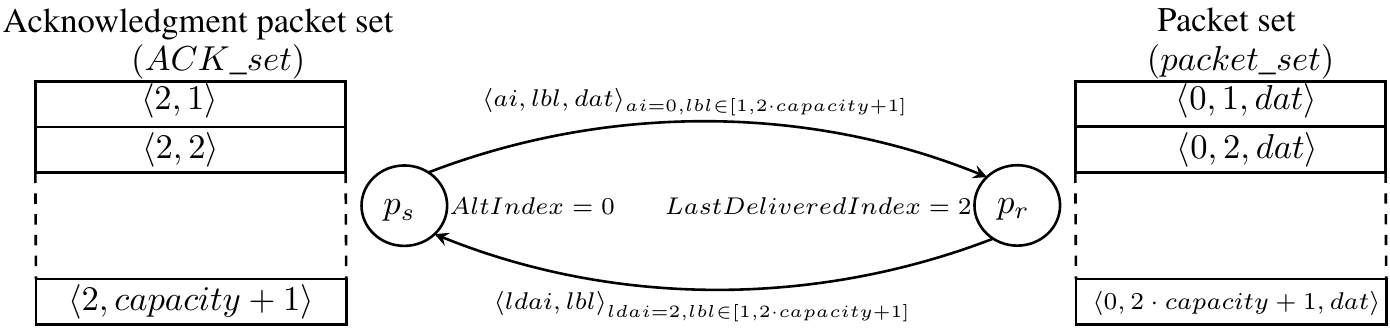}
			\end{center}
			
			\BBB\BBB\BBB
			The communication channels do not indicate to their receiving ends whether the transmitted packets were subject to duplication, omission or reordering. The algorithm facilitates the correct delivery of $m$ by letting $p_s$ send $(2 \cdot capacity+1)$ copies of the message $m = \langle dat \rangle$ to $p_r$, and requiring $p_r$ to receive $(2 \cdot capacity+1)$ packets, where the majority of them are copies of $m$. Namely, $p_s$ maintains an {\em alternating index}, $AltIndex \in [0,2]$, which is a counter that is incremented modulo $3$ every time $m$ is fetched and by that allow recovery from an arbitrary starting configuration. Moreover, $p_s$ transmits to $p_r$ a set of packets, $\langle ai, lbl, dat \rangle$, where $ai =AltIndex$, and $lbl$ are packet labels that distinguish this packet among all of $m$'s copies. The example illustrated above shows that when transmitting the packet set $\{ \langle 0,1,dat \rangle$, $\langle 0,2,dat \rangle$, $\ldots$, $\langle 0,2 \cdot capacity+1,dat \rangle \}$, the alternating index, $0$, distinguishes between this transmission set, and its predecessor set, which has the alternating index $2$, as well as the successor sets, which has the alternating index $1$. This transmission ends once $p_r$ receives a packet set, $\{ \langle 0, \ell, dat \rangle \}_{\ell \in [1, 2 \cdot capacity+1]}$, that is distinctly labeled by $\ell$ with respect to the alternating index $0$. (Note that when receiving a packet with a label that exists in the received packet set, the Receiver replaces the existing packet with the arriving one.) During legal executions, the set of received packets includes a majority of packets that have the same value of $dat$. When such a majority indeed exists, $p_r$ delivers $m = \langle dat \rangle$. After this decision, $p_r$ updates $LastDeliveredIndex \gets 0$ as the value of the last delivered alternating index.
			
			The correct packet transmission depends on the synchrony of $m$'s alternating index at the sending-side, and $LastDeliveredIndex$ on the Receiver side, as well as the packets that $p_r$ accumulates in $packet\_set_r$. The Sender repeatedly transmits this packet set until it receives $(capacity+1)$ distinctly labeled acknowledgment packets, $\langle ldai, lbl \rangle$, from the Receiver for which it holds that $ldai=AltIndex$. The Receiver acknowledges the Sender for each incoming packet, $\langle ai, lbl, dat \rangle$, using acknowledgment packet $\langle ldai, lbl \rangle$, where $ldai$ refers to the value, $LastDeliveredIndex$, of the last alternating index for which there was a receiving-side message delivery to the application layer. Thus, with respect to the above example, $p_s$ does not fetch another application layer message before it receives at least $(capacity+1)$ acknowledgment packets; each corresponding to one of the $(2 \cdot capacity+1)$ packets that $p_r$ received from $p_s$.
			%
			%
			On the receiving-side, $p_r$ delivers the message, $m = \langle dat \rangle$, from one of the $(capacity+1)$ (out of $(2\cdot capacity+1)$) distinctly labeled packets that have identical $dat$ and $ai$ values. After this delivery, $p_r$ assigns $ai$ to $LastDeliveredIndex$, resets its packet set and restarts accumulating packets, $\langle ai^\prime, lbl^\prime, dat^\prime \rangle$, for which $LastDeliveredIndex \neq ai^\prime$.

			
		\end{framed}
		\BBB
		
		\caption{\label{fig:smplalg}An \ems{ARQ algorithm} (first attempt)}

	\end{smaller}
	
	\BB\BB\BB
\end{figure*}


Let us consider the set of packets $X=\{ \langle ai, \ell, dat \rangle \}_{\ell \in [1, 2 \cdot capacity+1]}$ that $p_r$ receives during a legal execution, where $ai=0$, as in the example that appears in Figure~\ref{fig:smplalg}. We note that $X$ includes a majority of packets in $X$ that have the same value of $dat$, because the channel can add at most $capacity$ packets (due to channel faults, such as message duplication, or transient faults that occurred before the starting configuration), and thus $p_s$ has sent at least $(capacity+1)$ of these packets, i.e., the majority of the arriving packets to $p_r$ have originated from $p_s$, rather than the communication channel between $p_s$ and $p_r$ (due to channel faults or transient faults that occurred before the starting configuration). The protocol tolerates channel reordering faults, because the Sender fetches one message at a time, and since it does not fetch another before it receives an acknowledgment about the delivery of the current one. The protocol marks each packet with a distinct label in order to allow a packet selection that is based on majority in the presence of duplication faults.

The correctness proof considers: (1) asynchronous executions (while assuming fair communication but not execution fairness), and (2) the fact that the Receiver always acknowledges incoming packets. These acknowledgments can eventually arrive to the Sender, and hence the Sender fetches messages, $m = \langle dat \rangle$, repeatedly. Following $p_s$'s first fetch, $p_r$ can eventually receives $p_s$'s packets, $\langle ai, lbl, dat \rangle$, adopts  $p_s$'s alternating index, $ai=AltIndex$, records it in $LastDeliveredIndex$, delivers and acknowledges $m$, see the example in Figure~\ref{fig:smplalg}. We consider an execution in which $p_s$ changes its alternating index, $AltIndex$, as follows: $x$, $x+1$, $x+2$, $x$, $\ldots$ (where $x \in [0,2]$ and all operations are in modulo $3$). During this execution, $p_s$ receives acknowledgments that imply that $p_r$ changed $LastDeliveredIndex$ to $x+1$ and then to $x+2$. Moreover, the proof shows that between the acknowledgments with $LastDeliveredIndex=x+1$ and the acknowledgments with $LastDeliveredIndex=x+2$, the Sender does not send packets with alternating index $x$. Thus, $m$'s last delivery, with alternating index equals $x+1$, has to follow the reception of at least $(capacity+1)$, out of $(2\cdot capacity+1)$, distinctly labeled packets, $pckt=\langle x+1, \ast, dat \rangle$, in the sequence. This must be due to $m$'s (sending-side) fetch, $p_s$ transmission of $m$'s packets, $pckt=\langle x+1, \ast, dat \rangle$, from $p_s$ to $p_r$, and $m$'s (receiving-side) delivery.

\ems{The first attempt (Figure~\ref{fig:smplalg}) considers messages from the application layer, such that each message has a length of $ml$-bits. We note that by adding $t$ check symbols to the packet dataload, a Reed–Solomon code~\cite{reed1960polynomial} can correct up to $\lfloor t/2\rfloor$  erroneous symbols at unknown locations. One can calculate the value of $ml$ considers $n$ packets (each of $pl$-bits length) and the network capacity measured in the number of packets (of $pl$-bits length). If the capacity is no more than $t$ packets, there is a need to send $n=ml+2t+1$ packets to reconstruct the $ml$ packets. One can consider a method in which packets include the actual messages instead of a portion of every message (possibly using systematic codes). Note that also in this case $2t+1$ additional packets are needed to deal with $t$ packets (of $pl$-bits length) network capacity.}

The above first-attempt solution delivers each message exactly once in its (sending-order) while producing a large communication \ems{cost}. The proposed solution (Section~\ref{s:alg}) uses error correction codes and has a smaller \ems{communication cost}. It fetches a number of messages, $m$, on the sending-side. Then, it concurrently transmits them to the other end after transforming them into packets that are protected by error correction codes, and then delivering them at their sending order without omission, or duplication. We explain how to circumvent the difficulty that the communication channel can introduce up to capacity erroneous packets by considering the problem of having up to capacity erroneous bits in any packet.


\begin{figure*}
	\begin{smaller}
		
		\BBB\BBB\BB\B
		\begin{framed}
			\begin{center}
				\begin{minipage}{1.0\textwidth}
					\setlength{\unitlength}{1500sp}%
					\begingroup\makeatletter\ifx\SetFigFont\undefined%
					\gdef\SetFigFont#1#2#3#4#5{%
						\reset@font\fontsize{#1}{#2pt}%
						\fontfamily{#3}\fontseries{#4}\fontshape{#5}%
						\selectfont}%
					\fi\endgroup%
					\begin{center}
						
						\begin{picture}(9498,5934)(9841,-3695)
						
						\thinlines
						
						
						
						{\put(10050,-466){\line( 1, 0){2925}}
						}%
						{\put(10050,209){\line( 1, 0){2925}}
						}%
						
						{\multiput(10050,-1220)(0.00000,-90.47826){12}{\line( 0,-1){35.739}}
						}%
						{\multiput(12971,-1220)(0.00000,-90.39130){12}{\line( 0,-1){35.696}}
						}%

						{\put(10050,-2311){\line( 1, 0){2925}}
							\put(12991,-2311){\line( 0,-1){675}}
							\put(12991,-2986){\line(-1, 0){2925}}
							\put(10050,-2986){\line( 0, 1){675}}
						}%

						{\put(10050,-1186){\framebox(2925,2025){}}
						}%
						
						{\put(15080,-1186){\framebox(4230,2025){}}
						}%

						{\put(15080,-2986){\framebox(4230,675){}}
						}%
						{\put(15080,-466){\line( 1, 0){4230}}
						}%
						{\put(15080,209){\line( 1, 0){4230}}
						}%
						{}%
						{\put(9990,1964){\vector(-1, 0){  0}}
							\put(15076,1964){\vector( 1, 0){4230}}
						}%
						{\put(15100,1964){\vector(-1, 0){  0}}
							\put(10126,1964){\vector( 1, 0){2925}}
						}%
						{\put(15100,1964){\vector(-1, 0){  0}}
							\put(10126,1964){\vector( 1, 0){2925}}
						}%
						\thicklines
						\linethickness{0.5mm}
						{\put(15076,1739){\line( 0,-1){2925}}
						}%
						{\put(15076,1739){\line( 1, 0){675}}
						}%
						{\put(15076,-2311){\line( 0,-1){675}}
							\put(15076,-2986){\line( 1, 0){675}}
						}%

						{\put(19306,-1186){\line( 0, 1){2925}}
							\put(19306,1739){\line(-1, 0){630}}
							\put(18676,1739){\line( 0,-1){2925}}
						}%
						{\put(18676,-2311){\line( 0,-1){675}}
							\put(18676,-2986){\line( 1, 0){630}}
							\put(19306,-2986){\line( 0, 1){675}}
						}%
						{\put(15751,1739){\line( 0,-1){2925}}
						}%
						{\put(15751,1739){\line( 1, 0){675}}
						}%
						{\put(16426,1739){\line( 0,-1){2925}}
						}%
						{\put(15751,-2311){\line( 0,-1){675}}
							\put(15751,-2986){\line( 1, 0){675}}
							\put(16426,-2986){\line( 0, 1){675}}
						}%
						{\multiput(15751,-1186)(0.00000,-250.00000){5}{\line( 0,-1){125.000}}
						}%
						{\multiput(15076,-1186)(0.00000,-250.00000){5}{\line( 0,-1){125.000}}
						}%
						{\multiput(16426,-1186)(0.00000,-250.00000){5}{\line( 0,-1){125.000}}
						}%
						{\multiput(19306,-1186)(0.00000,-250.00000){5}{\line( 0,-1){125.000}}
						}%
						{\multiput(18676,-1186)(0.00000,-250.00000){5}{\line( 0,-1){125.000}}
						}%
						\thinlines
						{\put(16066,-3481){\vector( 0, 1){495}}
						}%
						{\put(18991,-3481){\vector( 0, 1){495}}
						}%
						{\put(15391,-3481){\vector( 0, 1){495}}
						}%
						{\put(18721,1289){\line( 1, 0){585}}
						}%
						{\put(15076,1289){\line( 1, 0){1350}}
						}%

						\linethickness{0.4mm}
						{\multiput(16400,1730)(114.39252,0.00000){20}{\line( 1, 0){ 37.196}}
						}%

						\put(9650,-2581){\makebox(0,0)[lb]{\smash{{\SetFigFont{8}{16.8}{\rmdefault}{\mddefault}{\itdefault}{pl}%
						}}}}
						\put(9780,-736){\makebox(0,0)[lb]{\smash{{\SetFigFont{8}{16.8}{\rmdefault}{\mddefault}{\updefault}{3}%
						}}}}
						\put(9780,-61){\makebox(0,0)[lb]{\smash{{\SetFigFont{8}{16.8}{\rmdefault}{\mddefault}{\updefault}{2}%
						}}}}
						\put(9780,614){\makebox(0,0)[lb]{\smash{{\SetFigFont{8}{16.8}{\rmdefault}{\mddefault}{\updefault}{1}%
						}}}}
						\put(18875,1469){\makebox(0,0)[lb]{\smash{{\SetFigFont{8}{16.8}{\rmdefault}{\mddefault}{\itdefault}{n}%
						}}}}
						\put(18831,1019){\makebox(0,0)[lb]{\smash{{\SetFigFont{8}{16.8}{\rmdefault}{\mddefault}{\itdefault}{ai}%
						}}}}
						\put(15286,1019){\makebox(0,0)[lb]{\smash{{\SetFigFont{8}{16.8}{\rmdefault}{\mddefault}{\itdefault}{ai}%
						}}}}
						\put(16000,1019){\makebox(0,0)[lb]{\smash{{\SetFigFont{8}{16.8}{\rmdefault}{\mddefault}{\itdefault}{ai}%
						}}}}
						\put(15300,1469){\makebox(0,0)[lb]{\smash{{\SetFigFont{8}{16.8}{\rmdefault}{\mddefault}{\updefault}{1}%
						}}}}
						\put(15975,1469){\makebox(0,0)[lb]{\smash{{\SetFigFont{8}{16.8}{\rmdefault}{\mddefault}{\updefault}{2}%
						}}}}
						\put(11476,2009){\makebox(0,0)[lb]{\smash{{\SetFigFont{8}{16.8}{\rmdefault}{\mddefault}{\itdefault}{ml}%
						}}}}
						\put(17100,2054){\makebox(0,0)[lb]{\smash{{\SetFigFont{8}{16.8}{\rmdefault}{\mddefault}{\itdefault}{n $>$ ml}%
						}}}}
						\put(14550,1469){\makebox(0,0)[lb]{\smash{{\SetFigFont{6}{16.8}{\rmdefault}{\mddefault}{\itdefault}{lbl}%
						}}}}
						\put(13800,1064){\makebox(0,0)[lb]{\smash{{\SetFigFont{6}{16.8}{\rmdefault}{\mddefault}{\itdefault}{AltIndex}%
						}}}}
						\put(13200,-1300){\makebox(0,0)[lb]{\smash{{\SetFigFont{8}{16.8}{\rmdefault}{\mddefault}{\updefault}{Error}%
						}}}}
						\put(13200,-1600){\makebox(0,0)[lb]{\smash{{\SetFigFont{8}{16.8}{\rmdefault}{\mddefault}{\updefault}{Correcting}%
						}}}}
						\put(13200,-1900){\makebox(0,0)[lb]{\smash{{\SetFigFont{8}{16.8}{\rmdefault}{\mddefault}{\updefault}{Encoding}%
						}}}}
						\put(14300,-3700){\makebox(0,0)[lb]{\smash{{\SetFigFont{6}{14.4}{\rmdefault}{\mddefault}{\updefault}{I$^{st}$ Packet}%
						}}}}
						\put(15796,-3700){\makebox(0,0)[lb]{\smash{{\SetFigFont{6}{14.4}{\rmdefault}{\mddefault}{\updefault}{II$^{ed}$ Packet}%
						}}}}
						\put(18300,-3700){\makebox(0,0)[lb]{\smash{{\SetFigFont{6}{14.4}{\rmdefault}{\mddefault}{\updefault}{$n^{th}$ Packet}%
						}}}}
						\put(13000,-2761){\makebox(0,0)[lb]{\smash{{\SetFigFont{6}{14.4}{\rmdefault}{\mddefault}{\updefault}{$pl^{th}$ Message}%
						}}}}
						\put(13000,-286){\makebox(0,0)[lb]{\smash{{\SetFigFont{6}{14.4}{\rmdefault}{\mddefault}{\updefault}{II$^{ed}$ Message}%
						}}}}
						\put(13000,389){\makebox(0,0)[lb]{\smash{{\SetFigFont{6}{14.4}{\rmdefault}{\mddefault}{\updefault}{I$^{st}$ Message}%
						}}}}
						
						\thicklines
						\linethickness{0.5mm}
						\put(13500,-905){\line( 1, 0){451}}
						\put(14100,-905){\vector( 1, 0){0}}
						
						\BBB\BBB\BBB\BBB\BBB\BBB\BBB\BBB\BBB\BBB
						\end{picture}
						\BBB\BBB\BBB\BBB\BBB\BBB\BBB\BBB\BBB\BBB
					\end{center}
					\BBB\BBB\BBB\BBB\BBB\BBB\BBB\BBB\BBB\BBB
					
					\BBB\BBB\BBB\BBB\BBB\BBB\BBB\BBB\BBB\BBB
				\end{minipage}
			\end{center}

			
			The (sending-side) encoder considers a batch of (same length) messages as a (bit) matrix, where each message (bit representation) is a matrix row. It transposes these matrices by sending the matrix columns as encoded data packets. Namely, the Sender fetches, $[m_j]_{j \in [1,pl]}$, a batch of $pl$ messages from the application layer each of length $ml$-bits, and calls the function $Encode([m_j]_{j \in [1,pl]})$. This function is based on an error correction code, $ecc$, that for $ml$-bits word, $m_j$, codes an $n$-bits words, $cm_j$, such that $cm_j$ can bear up to $capacity$ erroneous bits, i.e., $ecc$'s error threshold, $t_{ecc}$, is $capacity$. The function then takes these $pl$ (length $n$-bits) words, $cm_j$, and returns $n$ (length $pl$-bits) packet payloads, $[pyld_k]_{k \in [1,pl]}$, that are the columns of a bit matrix in which the $j$th row is $cm_j$'s bit representation, see the image above for illustration.
			
			The (receiver-side) uses the function $Decode([pyld^\prime_k]_{k \in [1,pl]})$, which requires $n$ packet payloads, and assumes that at most $capacity$ of them are erroneous packets that appeared in channel (due to transient faults that occurred before the starting configuration) rather than added to the channel by the Sender, or due to channel faults during the system execution, i.e., from $[pyld_k]_{k \in [1,pl]}$.
			%
			This function first transposes the arrived packet payloads, $[pyld^\prime_k]_{k \in [1,pl]}$ (in a similar manner to $Encode()$), before using $ecc$ for decoding $[m_j]_{j \in [1,pl]}$ and delivering the original messages to the Receiver's application layer.
			
			Namely when the Receiver accumulates $n$ distinct label packets, $capacity$ of the packets may be wrong or unrelated. However, since the $i^{th}$ packet, out of the $n$ distinctly labeled packets, encodes the $i^{th}$-bits of all the $pl$ encoded messages, if the $i^{th}$ packet is wrong, the decoder can still decode the data of the original $pl$ messages each of length $ml<n$. The $i^{th}$ bit in each encoded message may be wrong, in fact, capacity of packets maybe wrong yielding capacity of bits that may be wrong in each encoded message. However, due to the error correction, all the original $pl$ messages of length $ml$ can be recovered, so the Receiver can deliver the correct $pl$ messages in the correct order. Note that in this case, although the channel may reorder the packets, the labels maintain the sending-order, because the $i^{th}$ packet is labeled with $i$. In this proposed solution, the labels also facilitate duplication fault-tolerance, because the Receiver always holds at most one packet with label $i$, i.e., the latest.


			\vspace*{-\smallskipamount}
			\vspace*{-\smallskipamount}
			\vspace*{-\smallskipamount}
		\end{framed}
		\BBB
		
		\caption{\label{fig:pfftmis}Packet formation from messages}
		\BB\BB\BB
		
	\end{smaller}
\end{figure*}


\Subsection{Error correction codes for payload sequences}
%
Error correction codes~\cite{Moon} can mitigate bit-flip errors in (binary) words, where such words can represent payload in single data packet, or as we show here, can be used to help recover wrong words in a sequence of them. These methods use redundant information when encoding data, so that after the error occurrence, the decoding procedure will be able to recover the originally encoded data without errors. Namely, an error correction code $ec()$ encodes a payload $w$ (binary word) of length $wl$ with payload $c=ec(w)$ of length $cl$, where $cl > wl$. The payload $w$ can be later restored from $c^\prime$ as long as the Hamming distance between $c^\prime$ and $c$ is less than a known {\em error threshold}, $t_{ecc}$, where the Hamming distance between $c^\prime$ and $c$ is the smallest number of bits that one has to flip in $c$, in order to get $c^\prime$. 

Existing methods for error correction codes\ems{, such as Reed–Solomon error correction~\cite{reed1960polynomial},} can also be used for a sequence of packets and their payloads, see Figure~\ref{fig:pfftmis}. These sequences are encoded on the sender-side, and sent over a bounded capacity, omitting, duplicating and non-FIFO channel, before decoding them on the receiver-side. On that side, the originally encoded payload sequence is decoded, as long as the error threshold is not smaller than the channel capacity, i.e., $t_{ecc} \geq capacity$.
%
%
This method circumvents the issue of having up to $capacity$ erroneous packets by considering the problem of having up to $capacity$ erroneous bits in any packet. This problem is solved by using error correction codes to mask the erroneous bits. The proposed solution allows correct message delivery even though up to $capacity$ of packets are erroneous, i.e., packets that appeared in channel (due to transient faults that occurred before the starting configuration) rather than added to the channel by Sender, or due to channel faults during the system execution.

\Section{Self-Stabilizing \ems{Automatic Repeat reQuest (ARQ)} Algorithm}
\label{s:alg}
We propose an efficient $S^2ARQ$ algorithm that fetches a number of messages, $m$, and encodes them according to the method presented in Figure~\ref{fig:pfftmis}. The Sender then concurrently transmits $m$'s encoded packets to the receiving end until it can decode and acknowledge $m$. Recall that the proposed method for error correction can tolerate communication channels that, while in transit, omit, duplicate and reorder $m$'s packets, as well as add up to $capacity$ packets to the channel (due to transient faults that have occurred before the starting configuration rather than packets that the Sender adds to the channel during the system execution). We show how the Sender and the Receiver can use the proposed error correction method for transmitting and acknowledging $m$'s packets.

Reliable, ordered, and error-checked protocols in the transport layer, such as TCP/IP, often consider the delivery of a stream of octets between two network ends. The algorithm presented in Figure~\ref{fig:pfftmis} considers a transport layer protocol that repeatedly fetches another part of the stream. Upon each such fetch, the protocol breaks that part of the stream into $m$ sub-parts, and the protocol refers to $m$ sub-parts as the application layer messages. Note that the size of each such message can be determined by the (maximal) payload size of packets that the Sender transmits to the Receiver, because the payload of each transmitted packet needs to accommodate one of the $n$-bits words that are the result of transposing $m$ stream sub-parts. 
%

%

The $S^2ARQ$ algorithm extends the first attempt ARQ algorithm (Figure~\ref{fig:smplalg}), i.e., the Sender, $p_s$, transmits the packets $\langle ai, lbl, dat \rangle$, and the Receiver, $p_r$, acknowledges using the $\langle ai, lbl \rangle$ packets, where $ai \in [0,2]$ is the state alternating index, and $lbl$ are packet labels that are distinct among all of the packets that are associated with messages $m = \langle dat \rangle$. Moreover, it uses the notation of the proposed error correction method (Figure~\ref{fig:pfftmis}), i.e., the Sender fetches batches of $pl$ application layer messages of length $ml$-bits that are coded by $n$-bit payloads that tolerate up to $capacity$ erroneous bits. The Sender, $p_s$, fetches $pl$ (application layer) messages, $m = \langle dat \rangle$, encodes them into $n$ (distinctly labeled) packets, $\langle ai=AltIndex_s, lbl, dat \rangle$, according to the proposed error correction method (Figure~\ref{fig:pfftmis}), and repeatedly transmits these $n$ (distinctly labeled) packets to $p_r$ until $p_s$ receives from $p_r$ (at least) $(capacity+1)$ (distinctly labeled) acknowledgment packets $\langle ldai^\prime, lbl^\prime \rangle$, for which after convergence $ldai^\prime=AltIndex$. The Receiver repeatedly transmits the acknowledgment packets $\langle ldai^\prime, lbl^\prime \rangle$, which acknowledge the messages in the previous batch that it had delivered to its application layer that had the alternating index, $ai = LastDeliveredIndex$.
%
%
Note that the Receiver repeatedly sends  $(capacity+1)$ acknowledgment packets, as a response to the $n$ received packets, rather than a particular packet that has arrived.
Namely, the Receiver $p_r$ accumulates arriving packets, $\langle ai, lbl, dat \rangle$, whose alternating indexes, $ai$, is different from the last delivered one, $LastDeliveredIndex$. Moreover, once $p_r$ has $n$ (distinctly labeled) packets, which are $\{\langle ai, \ell, dat \rangle\}_{\ell \in [1,n]} :$ $ai$ $\neq$ $LastDeliveredIndex$, the Receiver $p_r$ updates $LastDeliveredIndex$ according to $ai$, as well as use the proposed error correction method (Figure~\ref{fig:pfftmis}) for decoding $m$ before delivering it.

Note that $p_s$ transmits to $p_r$ a set of $n$ (distinctly labeled with respect to a single alternating index) of $m$'s packets, that is, $m$'s packets, which the channel can omit, duplicate and reorder. Thus, once $p_r$ receives $n$ packets, $p_r$ can use the proposed error correction method (Figure~\ref{fig:pfftmis}) as long as their alternating index is different from the last delivered one, $LastDeliveredIndex$, because at least $(n-capacity)$ of these packets were sent by $p_s$. Similarly, $p_r$ transmits to $p_s$ a set of $(capacity + 1)$ (distinctly labeled with respect to a single alternating index) of $m$'s acknowledgment packets. Thus, once $(capacity + 1)$ of $m$'s acknowledgment packets (with $ldai$ matching to $AltIndex$) arrive at the sending-side, $p_s$ can conclude that at least one of them was transmitted by $p_r$, as long as their alternating index, $ldai$, is the same as the one used for $m$, $AltIndex$.

%
The correctness arguments show that eventually the system can reach an execution in which the Sender fetches a new message batch infinitely often, and the Receiver will deliver the messages fetched by the Sender before its fetches the next message batch. Thus, every batch of $pl$ fetched messages is delivered exactly once, because after delivery the Receiver resets its packet set and changes its $LastDeliveredIndex$ to be equal to the alternating index of the Sender. The Receiver stops accumulating packets from the Sender (that their alternating index is $LastDeliveredIndex$) until the Sender fetches the next message batch, and starts sending packets with a new alternating index. Note that the Sender only fetches new messages after it gets $(capacity + 1)$ distinctly labeled acknowledgments, $\langle ldai, lbl \rangle$ (that their alternating index, $ldai$, equals to $p_s$'s $AltIndex$). When the Receiver holds $n$ (distinctly labeled) packets out of which at most $capacity$ are erroneous ones, it can convert the packets back to the original messages, see (Figure~\ref{fig:pfftmis}).

\LinesNotNumbered \begin{algorithm*}[t!]
	\begin{smaller}
		
		{\bf Local variables:}
		
		$AltIndex \in [0,2]:$ state the current alternating index value
		
		$ACK\_set$: at most $(capacity + 1)$ acknowledgment set, where items contain \ems{last delivered alternating indexes and labels,} $\langle ldai, lbl \rangle$
		

		~\\
		
		{\bf Interfaces:}
		
		$Fetch(NumOfMessages)$ fetches $NumOfMessages$ messages from the application and returns them in an array of size $NumOfMessages$ according to their original order
		
		$Encode(Messages[])$ receives an array of messages of length $ml$ each, $M$, and returns a message array of identical size $M^\prime$, where message $M^\prime[i]$ is the encoded original $M[i]$, the final length of the returned $M^\prime[i]$ is $n$ and the code can bare $capacity$ mistakes
		
		~\\
		
		\nl{\bf Function $packet\_set()$} \nllabel{ln:Senderpacketsetfunction}
		\Begin{
			
			\nl		\lForEach {$(i, j) \in [1,n] \times [1,pl] $}{{\bf let} $data[i].bit[j] = messages[j].bit[i]$} \nllabel{ln:Sendercreatepacket}
			
			\nl		{\bf return} $\{ \langle AltIndex, i, data[i]  \rangle \}_{i \in [1,n]}$ 	
			
		}
		
		\nl {\bf Do forever} \nllabel{ln:SenderDoForever}
		\Begin
		{
			
			\nl \If{$( \{ AltIndex \} \times [1, capacity + 1 ]) \subseteq ACK\_set$} 
			{$(AltIndex, ACK\_set, messages) \gets ((AltIndex +1) \bmod 3, \emptyset, Encode(Fetch(pl)))$} \nllabel{ln:Senderencodefetch}
			
			\nl	\lForEach {$pckt \in packet\_set()$}{{\bf send} $pckt$} \nllabel{ln:Sendersend} 
		}
		\nl	{\bf Upon receiving $ACK = \langle ldai, lbl \rangle$} \nllabel{ln:Senderget}
		\Begin
		{
			
			\nl		\If{$ ldai= AltIndex \wedge lbl \in [1, capacity +1] $}{ \nllabel{ln:Sendergetsack}
				
				\nl			$ACK\_set \gets ACK\_set\cup \{ACK\}$ \nllabel{ln:Senderaddack}

				
				
				
			}
		}

		\caption{\label{alg:sender}Self-Stabilizing ARQ Algorithm (Sender $p_s$)}
		
	\end{smaller}
	
\end{algorithm*}

\Subsection{Detailed description}
The pseudo-code in Algorithms~\ref{alg:sender} and~\ref{alg:receiver} implements the proposed $S^2ARQ$ algorithm from the sender-side, and respectively, receiver-side. The two nodes, $p_s$ and $p_r$, are the Sender and the Receiver nodes respectively. The Sender algorithm consists of a do forever loop statement (lines~\ref{ln:SenderDoForever} to~\ref{ln:Senderencodefetch} of the Sender algorithm), where the Sender, $p_s$, assures that all the data structures comprises only valid contents. Namely, $p_s$ checks that the $ACK\_set_s$ holds packets with alternating index equal to the Sender's current $AltIndex_s$ and the labels are between $1$ and $(capacity + 1)$.


\begin{algorithm*}[t!]
	\begin{smaller}

		{\bf Local variables:}
		
		$LastDeliveredIndex \in [0,2]$: the alternating index value of the last delivered packets
		
		$packet\_set$: packets, $\langle ai, lbl, dat \rangle$, received, where $lbl \in [1,n]$ and $dat$ is data of size $pl$-bits

		~\\
		
		{\bf Interfaces:}
		
		$Decode(Messages[])$ receives an array of encoded messages, $M^\prime$, of length $n$ each, and returns an array of decoded messages of length $ml$, $M$, where $M[i]$ is the decoded $M^\prime[i]$. The code is the same error correction coded by the Sender and can correct up to $capacity$ mistakes
		
		$Deliver(messages[])$ receives an array of messages and delivers them to the application by the order in the array
		
		~\\
		
		{\bf Macro:}

		$ index(ind)=\{ \langle ind, \ast, \ast \rangle \in packet\_set \}$
		
		~\\
		
		\nl {\bf Do forever}
		\Begin
		{
			\nl \lIf{$\{ \langle ai, lbl \rangle : \langle ai, lbl, \ast \rangle \in
				packet\_set \} \not \subseteq \{ [0,2] \setminus \{ LastDeliveredIndex \} \} \times [1,n] \times \{ \ast \}  \vee$
				$(\exists \langle ai, lbl, dat \rangle \in packet\_set : \langle ai, lbl, \ast \rangle \in packet\_set \setminus \{ \langle ai, lbl, dat \rangle \}) \vee$
				$(\exists pckt = \langle \ast, \ast, data \rangle \in packet\_set : | pckt.data | \not = pl) \vee$ $1$ $<$ $| \{$ $AltIndex: n$ $\leq$ $| \{$ $\langle AltIndex,\ast,\ast \rangle  \in packet\_set \} |  \} | $}{$packet\_set \gets \emptyset$} \nllabel{ln:receiverdoforever}
			
			\nl              \If{$\exists~!~ind: ind \neq LastDeliveredIndex$ $\wedge$ $n \leq$ $| index(ind) |$}{ \nllabel{ln:receivergetsexactone}
				
				\nl 	\nllabel{ln:RmessagegenerationForEach}  \ForEach {$(i, j) \in [1,pl] \times [1,n] $}{
					
					\nl {\bf let}  $messages[i].bit[j]$ $=$ $data.bit[i] :\langle ind, j, data \rangle$  $\in$ $index(ind)$ \nllabel{ln:Rmessagegeneration}
					
				}				
				
				\nl    $(packet\_set, LastDeliveredIndex)  \gets (\emptyset, ind )$ \nllabel{ln:receiverchengeLDI}
				
				\nl    $Deliver$($Decode$($messages$)) \nllabel{ln:receiverdecodedeliver}
				
			}

			\nl \lForEach {$i \in [1,capacity+1]  $}{{\bf send} $\langle LastDeliveredIndex, i \rangle$} \nllabel{ln:receiversendACK}
			
		}
		\nl 	{\bf Upon receiving $pckt = \langle ai, lbl, dat \rangle$} \nllabel{ln:receivercheckpacketsetevent}
		\Begin
		{
			\nl 		\If{$\langle ai, lbl, \ast \rangle$ $\not \in$ $packet\_set$ $\wedge$ $ \langle ai, lbl \rangle$ $\in$ $(\{ [0,2]$ $\setminus$ $\{ LastDeliveredIndex \} \}$ $\times$ $[1,n])$ $\wedge$  $| dat |$ $=$ $pl$ }{ \nllabel{ln:receivercheckpacketset}
				\nl 		$packet\_set \gets packet\_set\cup \{pckt\}$ \nllabel{ln:recieverreceivingpacketset}

			}
			
		}
		
		\caption{\label{alg:receiver}Self-Stabilizing ARQ Algorithm (Receiver $p_r$)}
		
	\end{smaller}
	
\end{algorithm*}

In case any of these conditions is unfulfilled, the Sender resets its data structures (line~\ref{ln:Senderencodefetch} of the Sender algorithm). Subsequently, $p_s$ triggers the \textit{Fetch} and the \textit{Encode} interfaces (line~\ref{ln:Senderencodefetch} of the Sender algorithm). Before sending the packets, $p_s$ executes the $packet\_set()$ function (line~\ref{ln:Sendersend} of the Sender algorithm).

The Sender algorithm, also, handles the reception of acknowledgments $ACK_s = \langle ldai, lbl \rangle$ (line~\ref{ln:Senderget} of the Sender algorithm). Each packet has a distinct label with respect $m$'s message batch.  If $ACK_s = \langle ldai, lbl \rangle$ has the value of $ldai$ (last delivered alternating index) equals to $AltIndex$ (line~\ref{ln:Sendergetsack} of the Sender algorithm), the Sender $p_s$ stores $ACK_s$ in $ACK\_set_s$ (line~\ref{ln:Senderaddack} of the Sender algorithm). When $p_s$ gets such (distinctly labeled) packets $(capacity + 1)$ times, $p_s$ changes $AltIndex_s$, resets $ACK\_set_s$, and calls $Fetch()$ and $Encode()$ interfaces (line~\ref{ln:Senderencodefetch} of the Sender algorithm).

The Receiver algorithm executes at the Receiver side, $p_r$. The Receiver $p_r$ repeatedly tests $packet\_set_r$ (line~\ref{ln:receiverdoforever} of the Receiver algorithm), and assures that: (\textit{i}) $packet\_set_r$ holds packets with alternating index, $ai \in [0,2]$, except $LastDeliveredIndex_r$, labels ($lbl$) between $1$ and $n$ and data of size $pl$, and (\textit{ii}) $packet\_set_r$ holds at most one group of $ai$ that has (distinctly labeled) $n$ packets. When any of the aforementioned conditions do not hold, $p_r$ assigns the empty set to $packet\_set_r$.
When $p_r$ discovers that it has $n$ distinct label packets of identical $ai$ (line~\ref{ln:receivergetsexactone} of the Receiver algorithm), $p_r$ decodes the payloads of the arriving packets (line~\ref{ln:Rmessagegeneration} of the Receiver algorithm). Subsequent steps include the reset of $packet\_set_r$ and change of $LastDeliveredIndex_r$ to $ai$ (line~\ref{ln:receiverchengeLDI} of the Receiver algorithm). Next, $p_r$ delivers the decoded message (line~\ref{ln:receiverdecodedeliver} of the Receiver algorithm).
In addition, $p_r$ repeatedly acknowledges $p_s$ by $(capacity+1)$ packets (line~\ref{ln:receiversendACK} of the Receiver algorithm).


Node $p_r$ receives a packet $pckt_r = \langle ai, lbl, dat \rangle$, see line~\ref{ln:receivercheckpacketsetevent} (the Receiver algorithm). If $pckt_r$ has data ($dat$) of size $pl$-bits, an alternating index ($ai$) in the range of $0$ to $2$, excluding the $LastDeliveredIndex$, and a label ($lbl$) in the range of $1$ to $n$ (line~\ref{ln:receivercheckpacketset} of the Receiver algorithm), $p_r$ puts $pckt_r$ in $packet\_set_r$ (line~\ref{ln:recieverreceivingpacketset} of the Receiver algorithm).

\Section{Correctness}
\label{s:cor}
%
%
We define the set of legal executions, and how they implement the $S^2ARQ$ task (Section~\ref{s:sys}), before demonstrating that the Sender and the Receiver algorithms implement that task (theorems~\ref{th:closure} and~\ref{th:selfstab}). \ems{Our analysis considers Lamport's happened-before relation when demonstrating stabilization (Section~\ref{sec:timeComplexity}).}  

Given a system execution, $R$, and a pair, $p_s$ and $p_r$, of sending and receiving nodes, the $S^2ARQ$ task associates $p_s$'s sending message sequence $s_R= {{im}_{0}},$ ${{im}_{1}},$ ${{im}_{2}},$ $\ldots,$ ${{im}_{\ell}},$ $\ldots$, with $p_r$ delivered message sequence $r_R$ $=$ ${{om}}_{0},$ ${{om}_{1}},$ ${{om}_{2}},$ $\ldots,$ ${{om}_{\ell^\prime}},$ $\ldots$, see Section~\ref{s:sys}. The Sender algorithm encodes batch of messages, ${{im}_{\ell}}$, using an error correction method (Figure~\ref{fig:pfftmis}) into a packet sequence, $I$, that tolerates up to $capacity$ wrong packets (the Sender algorithm, line~\ref{ln:Senderencodefetch}). The Receiver decodes messages, ${{om}_{\ell^\prime}}$, from a packet sequence, $O$ (Receiver algorithm, line~\ref{ln:receiverdecodedeliver}), where every $n$ consecutive packets may have up to $capacity$ packets that were received due to channel faults rather than $p_s$ transmissions. Therefore, our definition of legal execution considers an unbounded suffix of input packets queue, $I = (im_x, im_{x+1},\ldots)$, which $p_s$ sends to $p_r$, and a $k$, such that the packet output suffix starts following the first $k-1$ packets, $O = (om_k, om_{k+1},\ldots)$, is always a prefix of $I$. Furthermore, a new packet is included in $O$ infinitely often.

\Subsection{Basic facts}
Throughout this section, we refer to $R$ as an execution of the Sender and the Receiver algorithms, where $p_s$ executes the Sender algorithm and $p_r$ executes the Receiver algorithm. Let $a_{s_\alpha}$ be the $\alpha^{th}$ time that the Sender is fetching a new message batch, i.e., executes line~\ref{ln:Senderencodefetch} (the Sender algorithm). Let $a_{r_\beta}$ be the $\beta^{th}$ time that the Receiver is delivering a message batch, i.e.,  executing line~\ref{ln:receiverdecodedeliver} (the Receiver algorithm). 

\remove{

\begin{theorem}[Liveness]
	\label{th:liveness}
	For every nice execution $R$, there exists an $R$'s prefix, $R^\prime$, that has $\BigO(n)$ [[@@ asynchronous rounds, @@]] and it includes at least one $a_{s_\alpha}$ step and least one $a_{r_\beta}$ step, where $n$ is the packet word length.
\end{theorem}

\begin{proof}
	By line~\ref{ln:Sendersend} (the Sender algorithm) and line~\ref{ln:receiversendACK} (the Receiver algorithm), node $p_i$ sends packets infinitely often to node $p_j$, where $i \in \{s,r\}$ and $j \in \{s,r\} \setminus \{ i \}$. Our system settings assume that when node $p_i$ sends a packet infinitely often to $p_j$ through the communication channel, node $p_j$ receives that packet infinitely often. This implies that within $\BigO(n)$ [[@@ asynchronous rounds @@]], the Receiver, $p_r$, receives all the $n$ packets in $packet\_set_s()$ and stores them all in $packet\_set_r$, cf. line~\ref{ln:recieverreceivingpacketset}, and thus the condition in line~\ref{ln:receivergetsexactone} (the Receiver algorithm) is satisfied. Therefore, $R^\prime$ includes at least one $a_{r_\beta}$ step. Moreover, the same argument implies that within $\BigO(n)$ [[@@ asynchronous rounds @@]], the Sender, $p_s$, receives $(capacity +1)$ acknowledgments from $p_r$ and stores them in the set $ACK\_set_s$, cf. line~\ref{ln:Senderaddack}, and thus the condition in line~\ref{ln:Senderencodefetch} (the Sender algorithm) is satisfied, where $capacity<n$. Therefore, $R^\prime$ includes at least one $a_{s_\alpha}$ step.
\end{proof}

} 

Lemmas~\ref{th:senderconvergence},~\ref{th:receiverconvergence} and~\ref{th:senderinter} are needed for the proof of Theorem~\ref{th:closure} and Theorem~\ref{th:selfstab}.

\begin{lemma}
	\label{th:senderconvergence}
	Let $c_{s_\alpha}(x)$ be the $x^{th}$ configuration between $a_{s_\alpha}$ and $a_{s_{\alpha+1}}$\ems{. Also, let} $ACK_\alpha =$ $\{ ack_\alpha(\ell)$ $\}_{\ell \in [1, capacity+1]}$ be a set of acknowledgment packets, where $ack_\alpha(\ell) = \langle s\_index_\alpha, \ell \rangle$.
	
	\begin{enumerate}
		\item
		\label{s1AI}
		For any given $\alpha >0$, there is a single index value, $s\_index_\alpha \in [0,2]$, such that for any $x >0$, it holds that $AltIndex_s = s\_index_\alpha$ in $c_{s_\alpha}(x)$.
		\item
		\label{sAI2LDI}
		Between $a_{s_\alpha}$ and $a_{s_{\alpha+1}}$ there is at least one configuration $c_{r_\beta}$, in which $LastDeliveredIndex_r = s\_index_\alpha$.
		\item
		\label{sMustSend}
		Between $a_{s_\alpha}$ and $a_{s_{\alpha+1}}$, the Sender, $p_s$, receives \ems{via} the channel from $p_r$ to $p_s$, the entire set, $ACK_\alpha$, of acknowledgment packets (each packet at least once), and between (the first) $c_{r_\beta}$ and $a_{s_{\alpha+1}}$ the Receiver must send at least one $ack_\alpha(\ell) \in ACK_\alpha$ packet, which $p_s$ receives, where $c_{r_\beta}$ is defined \ems{at item in~\ref{sAI2LDI} of this list.}
	\end{enumerate}
\end{lemma}

\begin{proof}
	We start by showing that $s\_index_\alpha$ exists (item~\ref{s1AI}) before showing that $c_{r_\beta}$ exists (item~\ref{sAI2LDI}) and that $p_s$ receives $ack_\alpha$ from $p_r$ between $a_{s_\alpha}$ and $a_{s_{\alpha+1}}$ (item~\ref{sMustSend}).
	
	\paragraph{Item~\ref{s1AI}  holds.}
	
	The value of $AltIndex_s = s\_index_\alpha$ is only changed in line~\ref{ln:Senderencodefetch} (the Sender algorithm). By the definition of $a_{s_\alpha}$, line~\ref{ln:Senderencodefetch} is not executed by any step between $a_{s_\alpha}$ and $a_{s_{\alpha + 1}}$. \ems{Thus, item~\ref{s1AI} holds.}
	
	
	\paragraph{Item~\ref{sAI2LDI}  holds.}
		
	We show that $c_{r_\beta}$ exists by showing that, between $a_{s_\alpha}$ and $a_{s_{\alpha+1}}$, there is at least one acknowledge packet, $\langle ldai, lbl \rangle$, that $p_r$ sends and $p_s$ receives, where $ldai = s\_index_\alpha$. This proves \ems{item~\ref{sAI2LDI},} because $p_r$'s acknowledgments are always sent with $ldai = LastDeliveredIndex_r$, see line~\ref{ln:receiversendACK} (the Receiver algorithm). \ems{Claim~\ref{thm:sAI2LDI} shows that item~\ref{sAI2LDI} holds.} 
	
	\begin{claim}
		\label{thm:sAI2LDI}
	\ems{Between $a_{s_\alpha}$ and $a_{s_{\alpha+1}}$, the Receiver $p_r$ sends at least one of the $ack_\alpha(\ell) \in ACK_\alpha$ packets that $p_s$ receives.}
	\end{claim}
	\begin{claimProof}
	 We show that $p_s$ receives, from the channel from $p_r$ to $p_s$, more than $capacity$ packets, i.e., the set $ACK_\alpha$. Since $capacity$ bounds the number of packets that, at any time, can be in the channel from $p_r$ to $p_s$, at least one of the $ACK_\alpha$ packets, say $ack_\alpha(\ell^\prime)$, must be sent by $p_r$ and received by $p_s$ between $a_{s_\alpha}$ and $a_{s_{\alpha+1}}$. This in fact proves that $p_r$ sends $ack_\alpha(\ell^\prime)$ after $c_{r_\beta}$. 
	\end{claimProof}
	
	\paragraph{Item~\ref{sMustSend} holds.}
	
	In order to demonstrate that $p_s$ receives the set $ACK_\alpha$, we note that $ACK\_set = \emptyset$ in configuration $c_{s_\alpha}(1)$, which immediately follows $a_{s_{\alpha}}$, see line~\ref{ln:Senderencodefetch} (the Sender algorithm). The Sender tests the arriving acknowledgment packet, $ack_\alpha$, in line~\ref{ln:Sendergetsack} (the Sender algorithm). \ems{This test asserts} $ack_\alpha$'s label to be in the range of $[1, capacity+1]$, and that they are of $ack_\alpha$'s form. Moreover, \ems{$p_s$} counts that $(capacity+1)$ different packets are added to $ACK\_set$ by adding them to $ACK\_set$, and not executing line~\ref{ln:Senderencodefetch} (the Sender algorithm) before at least $(capacity+1)$ distinct packets are in $ACK\_set$. The rest of item~\ref{sMustSend}'s proof is by Claim~\ref{thm:sAI2LDI}.
	
	This ends Lemma~\ref{th:senderconvergence}'s proof. 
\end{proof}

\begin{lemma}
	
	\label{th:receiverconvergence}
	Let $c_{r_\beta}(y)$ be the $y^{th}$ configuration between $a_{r_\beta}$ and $a_{r_{\beta+1}}$, and $PACKET_\beta(r\_index^\prime_\beta) = \{ packet_{\beta}(\ell, r\_index^\prime_\beta) \}_{\ell \in [1, n]}$ be a packet set, which could be a subset of the Receiver's $packet\_set_r$, where $packet_{\beta}(\ell, r\_index^\prime_\beta) = \langle r\_index^\prime_\beta, \ell, \ast \rangle$.

	\begin{enumerate}
		\item
		\label{r1AI}
		For any given $\beta >0$, there is a single index value, $r\_index_\beta \in [0,2]$, such that for any $y >0$, it holds that $LastDeliveredIndex_r = r\_index_\beta$ in configuration $c_{r_\beta}(y)$.
		\item
		\label{rAI2LDI}
		Between $a_{r_\beta}$ and $a_{r_{\beta+1}}$ there is at least one configuration, $c_{s_\alpha}$, such that $AltIndex_s \neq r\_index_\beta$.
		\item
		\label{rMustSend}
		There exists a single $r\_index^\prime_\beta \in [0,2] \setminus \{ r\_index_\beta \}$, such that the Receiver, $p_r$, receives all the packets in $PACKET_\beta(r\_index^\prime_\beta)$ at least once between $c_{s_\alpha}$ and $a_{r_{\beta+1}}$, where $c_{s_\alpha}$ is defined in \ems{item~\ref{rAI2LDI} of this list}, and at least ($n-capacity>0$) of them are sent by the Sender $p_s$ between $a_{r_\beta}$ and $a_{r_{\beta+1}}$.
	\end{enumerate}
\end{lemma}

\begin{proof}
	We begin the proof by showing that $r\_index_\beta$ exists (item~\ref{r1AI}) before showing that $c_{s_\alpha}$ exists (item~\ref{rAI2LDI}) and that $p_r$ receives the packets $packet_{\beta,r\_index^\prime_\beta}(\ell)$ from $p_s$ (item~\ref{rMustSend}).
	
	\paragraph{Item~\ref{r1AI} holds.}
	
	The value of $LastDeliveredIndex_r = r\_index_\beta$ is only changed in line~\ref{ln:receiverchengeLDI} (the Receiver algorithm). By the definition of $a_{r_\beta}$, line~\ref{ln:receiverchengeLDI} is not executed by any step between $a_{r_\beta}$ and $a_{r_{\beta + 1}}$. Therefore, for any given $\beta$, there is a single index value, $r\_index_\beta \in [0,2]$, such that for any $y >0$, it holds that $LastDeliveredIndex_r = r\_index_\beta$ in $c_{s_\beta}(y)$. \ems{Thus, item~\ref{r1AI} holds.}
	
	\paragraph{Item~\ref{rAI2LDI} holds.}
	
	We show that $c_{s_\alpha}$ exists by \ems{observing} that the Receiver, $p_r$, receives all the packets in \ems{$PACKET_\beta(r\_index^\prime_{\beta+1}):r\_index^\prime_{\beta+1}\neq r\_index_\beta$ via} the channel from $p_s$ to $p_r$, (each at least once) between $a_{r_{\beta}}$ and $a_{r_{\beta+1}}$. \ems{Recall that} $capacity$ bounds the number of packets that can be in the channel from $p_s$ to $p_r$, at any time. Hence, a subset, $S_\beta(r\_index^\prime_{\beta+1}) \subseteq PACKET_\beta(r\_index^\prime_{\beta+1})$, of at least ($(n-capacity)>0$) packets must be sent by $p_s$ between $a_{r_\beta}$ and $a_{r_{\beta+1}}$. This in fact proves that $p_s$ sends $S_\beta(r\_index^\prime_{\beta+1})$ after (the first) $c_{s_\alpha}$, because $p_s$ uses \ems{some $r\_index^{\prime\prime}_{\beta+1} \in  [0,2]$} as the alternating index for all the packets in $S_\beta(r\_index^\prime_{\beta+1})$, see \ems{lines~\ref{ln:Senderencodefetch} to~\ref{ln:Sendersend} (the Sender algorithm) as well as line~\ref{ln:receivercheckpacketset} (the Receiver algorithm). Moreover, by the definition of} function $packet\_set()$, as well as by the \ems{item~\ref{r1AI},} we have \ems{$r\_index^{\prime\prime}_{\beta+1} = r\_index^{\prime}_\beta\neq r\_index_\beta$. Thus, item~\ref{rAI2LDI} holds.}

	\paragraph{Item~\ref{sMustSend} holds.}
		
	~~~~\ems{We} show that, between $a_{r_\beta}$ and $a_{r_{\beta+1}}$, the Receiver $p_r$ receives packets, $packet_{\beta,r\_index^\prime_\beta}(\ell) \in PACKET_\beta(r\_index^\prime_\beta)$, with $n$ distinct labels from the channel from $p_s$ to $p_r$ . We note that $packet\_set_r = \emptyset$ in the configuration $c_{r_\beta}(1)$, which immediately follows $a_{r_{\beta}}$, see line~\ref{ln:receiverchengeLDI} (the Receiver algorithm). The Receiver tests the arriving packets, $packet_{\beta,r\_index^\prime_\beta}(\ell)$, in line~\ref{ln:receivercheckpacketset} (the Receiver algorithm). \ems{This test asserts} $packet_{\beta,r\_index^\prime_\beta}(\ell)$'s label to be in the range of $[1, n]$, $ packet_{\beta,r\_index^\prime_\beta}(\ell)$'s index to be different from $LastDeliveredIndex_r$ and that they are of $packet_{\beta,r\_index^\prime_\beta}(\ell)$'s form. Moreover, \ems{$p_r$} counts that $n$ packets with alternating index different from $LastDeliveredIndex_r$ and $n$ distinct labels are added to $packet\_set_r$ by not executing lines~\ref{ln:RmessagegenerationForEach} to~\ref{ln:receiverdecodedeliver} (the Receiver algorithm) before at least $n$ distinct labels are in $packet\_set_r$. \ems{Thus, item~\ref{sMustSend} holds.}
	
	This ends Lemma~\ref{th:receiverconvergence}'s proof. 
\end{proof}

Lemma~\ref{th:senderinter} borrows notation from lemmas~\ref{th:senderconvergence} and~\ref{th:receiverconvergence}.

\begin{lemma}
	\label{th:senderinter}
	\label{th:receiverinter}
	Suppose that $\alpha, \beta > 2$. Between $a_{s_\alpha}$ and $a_{s_{\alpha+1}}$, the Receiver takes at least one $a_{r_\beta}$ step, and that between $a_{r_\beta}$, and $a_{r_{\beta+1}}$, the Sender takes at least one $a_{s_{\alpha}}$ step. Moreover, equations~$\ref{eq:sindalpha}$ to $\ref{eq:sindbetaalpha}$ hold.
	%
	%
	\begin{eqnarray}
	\label{eq:sindalpha}
	s\_index_{\alpha+1} &=& s\_index_{\alpha} + 1 \bmod 3\\
	\label{eq:sindbeta}
	r\_index_{\beta+1} &=& r\_index_{\beta} + 1 \bmod 3\\
	\label{eq:sindalphabeta}
	r\_index_{\beta} &=& s\_index_{\alpha} \\
	\label{eq:sindbetaalpha}
	s\_index_{\alpha+1} &=& r\_index_{\beta} + 1 \bmod 3
	\end{eqnarray}
	
\end{lemma}

\begin{proof}
~	
\paragraph{Between $\mathbf{a_{s_\alpha}}$ and $\mathbf{a_{s_{\alpha + 1}}}$, there is at least one $\mathbf{a_{r_\beta}}$ step.}
	By Lemma~\ref{th:senderconvergence} and line~\ref{ln:Senderencodefetch} (the Sender algorithm), in any configuration, $c_{s_1}(x)$, that is between $a_{s_1}$ and $a_{s_{2}}$, the Sender is using a single alternating index, $s\_index_1$, and in any configuration, $c_{s_2}(x)$, that is between $a_{s_2}$ and $a_{s_{3}}$, the Sender is using a single alternating index, $s\_index_2$, such that $s\_index_2 = s\_index_1 + 1 \bmod 3$. Similarly, consider configuration, $c_{s_\alpha}(x)$, that is between $a_{s_\alpha}$ and $a_{s_{\alpha+1}}$ and conclude equations~$\ref{eq:sindalpha}$ and~$\ref{eq:sindalphabeta}$, cf. \ems{line~\ref{ln:Senderencodefetch} as well as} item $\ref{s1AI}$, and respectively, $\ref{sAI2LDI}$ of Lemma~\ref{th:senderconvergence}.
	
	Lemma~\ref{th:senderconvergence} also shows that for $\alpha \in \{1,2, \ldots \}$, there are configurations, $c_{r_{\beta}}$, in which $s\_index_\alpha=LastDeliveredIndex_r$. This implies that between $a_{s_\alpha}$ and $a_{s_{\alpha + 1}}$, the Receiver changes the value of $LastDeliveredIndex_r$ at least once, where \ems{$\alpha \in \{1,2, \ldots\}$.} Thus, by $a_{r_\beta}$'s definition and line~\ref{ln:receiverchengeLDI} (the Receiver algorithm), there is at least one $a_{r_\beta}$ step between $a_{s_\alpha}$ and $a_{s_{\alpha + 1}}$.
	
\paragraph{Between $\mathbf{a_{r_\beta}}$ and $\mathbf{a_{r_{\beta + 1}}}$, there is at least one $\mathbf{a_{s_\alpha}}$ step.}
	By \ems{the proof of} Lemma~\ref{th:receiverconvergence} and line~\ref{ln:receiverchengeLDI} (the Receiver algorithm), in any configuration, $c_{r_1}(y)$, that is between $a_{r_1}$ and $a_{r_{2}}$, the Receiver is using a single $LastDeliveredIndex_r$ \ems{value, denoted by} $r\_index_1$, and in any configuration, $c_{r_2}(y)$, that is between $a_{r_2}$ and $a_{r_{3}}$, the Receiver is using a single $LastDeliveredIndex_r$  \ems{value, denoted by} $r\_index_2$, such that $r\_index_2 = r\_index_1 + 1 \bmod 3$. In a similar manner, we consider configuration, $c_{r_\beta}(y)$, that is between $a_{r_\beta}$ and $a_{r_{\beta + 1}}$ and conclude equations~$\ref{eq:sindbeta}$ and~$\ref{eq:sindbetaalpha}$, cf. \ems{line~\ref{ln:receiverchengeLDI} (the Receiver algorithm) as well as} item $\ref{r1AI}$, and respectively, $\ref{rAI2LDI}$ of Lemma~\ref{th:receiverconvergence}.
	
	Lemma~\ref{th:receiverconvergence} also shows that for $\beta \in \{1,2, \ldots\}$, there are configurations, $c_{s_{\alpha}}$, in which $AltIndex_s \neq r\_index_\beta$. This implies that between $a_{r_\beta}$ and $a_{r_{\beta + 1}}$, the Sender changes the value of $AltIndex_s$ at least once. Thus, by $a_{s_\alpha}$'s definition, there is at least one $a_{s_\alpha}$ step between $a_{r_\beta}$ and $a_{r_{\beta + 1}}$.	
\end{proof}

\Subsection{Closure}
\label{sec:closure}
\ems{The closure property proof uses the following definition of a safe configuration, which considers all the alternating indices that are in a given configuration, $c$, such as the packet set indices, $\langle ind, lbl \rangle \in \{ \langle AltIndex, lbl \rangle:\langle AltIndex, lbl, dat \rangle \in packet\_set \}$, and the indices of the acknowledgment packet set, $\langle ind, lbl \rangle \in ACK\_set$ $ =$ $\{ \langle AltIndex, lbl \rangle \}$. } 

\ems{Given $X \in AI$, we define $index(ind, X)=\{ \langle ind, lbl \rangle : \langle ind, lbl,\bullet \rangle \in X  \}$, where $AI=\{ packet\_set, ACK\_set, channel_{s,r}, channel_{r,s}\}$ as well as $channel_{s,r}$ and $channel_{r,s}$ are the communication channel sets from the Sender to the Receiver, and respectively, from the Receiver to the Sender. We denote by $\{ 0^{\kappa_0}, 1^{\kappa_1}, 2^{\kappa_2} \}_X$ the fact that in configuration $c$, it holds $\forall i \in [0,2] : \kappa_i=|index(i, X)|$, where $X \in \{ packet\_set, ACK\_set \}$ refers to the values of $packet\_set$ and $ACK\_set$ in $c$.} 

\ems{ We consider the alternating index sequence, $ais$, stored in $AltIndex_s, \{ 0^{\kappa_0}, 1^{\kappa_1}, 2^{\kappa_2} \}_{channel_{s,r}}$, $\{ 0^{\kappa_0}, 1^{\kappa_1}, 2^{\kappa_2} \}_{packet\_set_r}$, $LDI_r$, $\{ 0^{\kappa_0},$ $1^{\kappa_1},$ $2^{\kappa_2} \}_{channel_{r,s}}$, and $\{ 0^{\kappa_0},$ $1^{\kappa_1},$ $2^{\kappa_2} \}_{ACK\_set_s}$ in this order, where $LDR_r$'s value is $LastDeliveredIndex_r$ as well as $channel_{s,r}$ and $channel_{r,s}$ are the communication channel sets from the Sender to the Receiver, and respectively, from the Receiver to the Sender. We show that a configuration, $c$, in which $ais = y$, $\{ \bullet^\ast \}$, $\{ z^{\kappa_z} \}_{z \in [0,2] \setminus \{ y \}}$, $y$, $\{ \bullet^\ast \}$, $\{ y^{capacity+1} \}$ is a safe configuration (Theorem~\ref{th:closure}), where $\ast$ denotes any finite value, $y \in [0,2]$ and $capacity \geq (\sum_{z \in [0,2] \setminus \{ y \}}  \kappa_z)$. Namely, $c$ is a safe configuration, which starts an execution that is in $LE_{S^2ARQ}$.}

\begin{theorem}[$S^2ARQ$ closure]
	\label{th:closure}
	%
	%
	Suppose that in $R$'s first configuration, $c$, it holds that $ais = y$, $\{ \bullet^\ast \}$, $\{ z^{\kappa_z} \}_{z \in [0,2] \setminus \{ y \}}$, $y$, $\{ \bullet^\ast \}$, $\{ y^{capacity+1} \}$ is a safe configuration, where $y \in [0,2]$ and $capacity \geq (\sum_{z \in [0,2] \setminus \{ y \}}  \kappa_z)$. Then, $c$ is safe.
	%
	%
\end{theorem}
\begin{proof}
	The correctness proof shows \ems{that} after configuration $c$, the system reaches configurations in which: (1) the Sender, $p_s$, increments its alternating index and starts transmitting a new message batch, $m$, (2) $p_r$, receives between $(n-capacity)$ and $n$ of $m$'s packets (with its alternating index from the previous step), and (3) $p_s$ receives at least one acknowledgment (with the alternating index from the first step) in which the $p_r$ acknowledges $m$'s packets. The proof shows that this is how $p_s$ and $p_r$ exchange messages and alternative indices. Therefore, $c$ starts a legal execution. For the sake of a simple presentation, we assume that $y=0$. \ems{Generality is not lost since the same arguments hold for any $y \in [0,2]$.}
	
	\medskip
	
	\ems{\textbf{(i)} The proof's focus now move to the action taken by the Sender.}
	In $c$, $p_s$'s state satisfies the condition $(ACK\_set = \{ AltIndex \} \times [1, capacity$ $+ 1 ])$ of line~\ref{ln:Senderencodefetch} (the Sender algorithm). Therefore, $p_s$ increments $AltIndex$ ($\bmod~3$), empties $ACK\_set_s$ and fetches a new batch of $pl$ messages, $m$, that it needs to sent to $p_r$. Thus, the system reaches configuration $c^{\prime}$ in which $ais$ $=$ $1$, $\{ \bullet^\ast \},$ $\{ \bullet^\ast \}$, $0$, $\{ \bullet^\ast \}$, $\{ \}$.
	
	\medskip
	
	\ems{(ii) The proof continues by reviewing the assertions made on the sender-side.}
	Note that by lines~\ref{ln:Senderencodefetch} to~\ref{ln:Sendersend} (the Sender algorithm), $p_s$ does not stop sending $m$'s packets with alternating indices $1$ until $ACK\_set_s$ has $(capacity+1)$ packets with the alternating index that is equal to $AltIndex_s=1$. Until that happens,
	lines~\ref{ln:Sendergetsack} and~\ref{ln:Senderaddack} (the Sender algorithm) implies that $p_s$ accepts acknowledgments that their alternating index is $AltIndex_s=1$, i.e., $ais$ $=$ $1$, $\{ 1^\ast, \bullet^\ast \},$ $\{ \bullet^\ast \}$, $0$, $\{ \bullet^\ast \}$, $\{ 1^\ast \}$.
	
	\medskip
	
	\ems{(iii) The proof's focus now moves to the action taken by the Receiver.}
	By line~\ref{ln:receivergetsexactone}, as well as lines~\ref{ln:receivercheckpacketsetevent} and~\ref{ln:recieverreceivingpacketset} (the Receiver algorithm), $p_r$ does not stop accepting $m$'s packets, which have alternating indices $1$, until $packet\_set_r$ has $n$ packets with \ems{the alternating index $1$, which} is different from the value of $LastDeliveredIndex_r=0$. Recall that the communication channel set, $channel_{s,r}$, from the Sender to the Receiver contains at most $capacity$ packets. Therefore, once $p_r$ has $n$ packets in $packet\_set$ (with alternating index $ai \neq 0$), $p_r$ must have received at least $(n-capacity)$ of these packets from $p_s$. Thus, the system reaches configuration $c^{\prime\prime}$ in which $ais$ $=$ $1$, $\{ 1^\ast, \bullet^\ast \},$ $\{ 1^n, 2^{\kappa_2} \}$, $0$, $\{ \bullet^\ast \}$, $\{  1^\ast \}$, where check this @@ $capacity \geq \kappa_2$.
	
	By lines~\ref{ln:receivergetsexactone} to~\ref{ln:receiverdecodedeliver} (the Receiver algorithm), $p_r$ empties $packet\_set_r$, updates $LastDeliveredIndex_r$ with the alternating index, $1$, of $m$'s packets, before decoding and delivering the messages encoded by $packet\_set_r$, as well as starting to send acknowledgements with the alternating index $LastDeliveredIndex_r = 1$, see line~\ref{ln:receiversendACK} (the Receiver algorithm). Thus, the system reaches configuration $c^{\prime\prime\prime}$ in which $ais$ $=$ $1$, $\{ 1^\ast, \bullet^\ast \},$ $\{  \}$, $1$, $\{ 1^\ast, \bullet^\ast \}$, $\{  1^\ast \}$.
	
	\medskip
	
	\ems{(iv) The proof's focus now shifts from the Receiver back to the Sender.}
	By line~\ref{ln:receiversendACK} (the Receiver algorithm) and lines~\ref{ln:Sendergetsack} and~\ref{ln:Senderaddack} (the Sender algorithm), $p_r$ keeps on acknowledging $m$'s packets until $p_s$ receives $(capacity+1)$ packets of acknowledgment from $p_r$. Thus, the system reaches configuration $c^{\prime\prime\prime\prime}$ in which $ais$ $=$ $1$, $\{ 1^\ast, \bullet^\ast \},$ $\{ 0^\ast, 2^\ast \}$, $1$, $\{ 1^\ast, \bullet^\ast \}$, $\{  1^{(capacity+1)} \}$.
	
	Note that, for $y=1$, this lemma claims that $c^{\prime\prime\prime\prime}$ is safe. Moreover, since we started in a configuration in which the communication channel sets from the Sender to the Receiver, and the Receiver to the Sender had no $n>capacity$, and respectively, $(capacity+1)$ packets with the alternating index $1$ exist, the Sender must have received at least one acknowledgment for $m$'s packet with the alternating index $1$ only after the Receiver receives at least one of for $m$'s packet with alternating index $1$, which happened after $p_s$ had fetched the batch messages of $m$ and incremented its alternating index to $1$. Therefore, $c$ starts a legal execution.
\end{proof}

\Subsection{Convergence}

Theorem~\ref{th:selfstab} uses the term safe configuration, which is a configuration that can only be followed by a legal execution (Section~\ref{sec:closure}). The proof of this theorem borrows notation from lemmas~\ref{th:senderconvergence} and~\ref{th:receiverconvergence} \ems{as well as Theorem~\ref{th:closure}. The proof itself is facilitated by Lemma~\ref{th:senderinter}.} 

\begin{theorem}[$\mathbf{S^2ARQ}$ convergence]
	\label{th:selfstab}
	\ems{(i)} Within four appearances of $a_{s_\alpha}$ steps and four of $a_{r_\beta}$ steps in $R$, the system reaches a safe configuration. \ems{(ii) The length of the longest chain of Lamport's happened-before relation is $8$ for any execution $R$.} 
\end{theorem}

\begin{proof}
	%
~~	\paragraph{Invariant (i) holds.}
	Let $c_{s_\alpha}(1)$ and $c_{r_\beta}(1)$ be the first configurations between $a_{s_\alpha}$ and $a_{s_{\alpha+1}}$, and respectively, between $a_{r_\beta}$ and $a_{r_{\beta+1}}$. Moreover, $s\_index_{\alpha}$ and $r\_index_{\beta}$ are $AltIndex_s$'s value in $c_{s_\alpha}(1)$, and respectively, $LastDeliveredIndex_r$'s value in $c_{r_\beta}(1)$.
	\ems{Thus, we can say} that $ais = s\_index_{\alpha}$, $\{ \bullet^\ast \}$, $\{ \bullet^\ast \}$, $r\_index_{\beta}$, $\{ \bullet^\ast \}$, $\{ \bullet^\ast  \}$ \ems{in $c_{s_\alpha}(1)$.}
	We show that, within \ems{the lemma's} four appearances of $a_{s_\alpha}$ steps and four of $a_{r_\beta}$ steps in $R$, the system reaches a configuration in which $ais$ $=$ $r\_index_{\beta+1}$, $\{ \bullet^\ast \}$, $\{ z^{\kappa_z} \}_{z \in [0,2] \setminus \{ r\_index_{\beta+1} \}}$, $r\_index_{\beta+1}$, $\{ \bullet^\ast \}$, $\{ r\_index_{\beta+1}^{capacity+1} \}$, \ems{where $capacity \geq (\sum_{z \in [0,2] \setminus \{ r\_index_{\beta+1} \}}  \kappa_z)$. Recall that Theorem~\ref{th:closure} says that this condition implies a safe configuration.}
	%
	
	%
	By Lemma~\ref{th:senderinter}, $\forall \alpha, \beta >2$ it holds that between $c_{s_\alpha}(1)$ and $c_{s_{\alpha+1}}(1)$ the system execution includes $c_{r_\beta}(1)$ in which Equation~$\ref{eq:sindalphabeta}$ holds. Namely, $\forall \alpha, \beta >3$, it holds that $r\_index_{\beta+1}=s\_index_{\alpha}$, and thus, in $c_{r_\beta}(1)$ it holds that $ais = r\_index_{\beta+1}$, $\{ \bullet^\ast \}$, $\{ z^{\kappa_z} \}_{z \in [0,2] \setminus \{ r\_index_{\beta+1} \}}$, $r\_index_{\beta+1}$, $\{ \bullet^\ast \}$, $\{ r\_index_{\beta+1}^{\kappa_{r\_index_{\beta+1}}} \}$. 
	
	The rest of the proof \ems{of invariant (i)} demonstrates that $capacity+1={\kappa_{r\_index_{\beta+1}}}$ and $capacity \geq (\sum_{z \in [0,2] \setminus \{ r\_index_{\beta+1} \}}  \kappa_z)$. \ems{The proof of the former is followed} by arguments similar to the ones in the proof of Theorem~\ref{th:closure}, \ems{cf. item (iv). For the proof of the latter, recall the above we showed that between $a_{s_\alpha}$ and $a_{s_{\alpha+1}}$ the sender sends at least one packet with the alternating index $s\_index_{\alpha}=r\_index_{\beta+1}$ over $channel_{s,r}$. The $capacity \geq (\sum_{z \in [0,2] \setminus \{ r\_index_{\beta+1} \}}  \kappa_z)$ is implied by the bound on the channel. Thus, the proof of invariant (i).} 

~~	\paragraph{Invariant (ii) holds.} \ems{From the proof of invariant (i), we observe that the appearance of steps $a_{s_\alpha}, a_{s_{\alpha+1}},a_{s_{\alpha+2}}$ and $a_{s_{\alpha+3}}$ in $R$ is interleaved eventually with the appearance of steps $a_{r_\beta},a_{r_{\beta+1}},a_{r_{\beta+2}}$ and $a_{r_{\beta+4}}$. Moreover, each such $a_{s_\alpha}$ step (and $a_{r_\beta}$ step) sends messages that are sufficient and necessary for the execution of a single $a_{r_\beta}$ step ($a_{s_\alpha}$ step, respectively). Thus, these 8 steps imply that the length of the longest chain of Lamport's happened-before relation is $8$ for any execution $R$ during which the system recover from the occurrence of the last transient fault.}
\end{proof}

\Section{Conclusions}
We have proposed self-stabilizing \ems{Automatic Repeat reQuest (ARQ)} algorithms for bounded capacity computer networks. The proposed algorithms inculcate error correction methods for the delivery of messages to their destination without omissions, duplications or reordering. We consider two nodes, one as the Sender and the other as the Receiver. In many cases, however, two communicating nodes may act both as senders and receivers simultaneously. In such situations, acknowledgment piggybacking may reduce the \ems{communication cost} needed to cope with up to $capacity$ erroneous packets that may exist in each direction, from the Sender to the Receiver {\em and} from the Receiver to the Sender (i.e., the erroneous packets that can appear in channel due to transient faults that occurred before the starting configuration). Using piggybacking, the \ems{communication cost} is similar in both directions. The obtained \ems{communication cost} is proportional to the ratio between the number of bits in the original message, and the number of bits in the coded message, which is a code that withstands $capacity$ corruptions. Thus, for a specific {\em capacity} and assuming encoding efficiency, this \ems{communication cost} becomes smaller as the message length grows.

\section*{Acknowledgments}
The work of the first author, S. Dolev, was partially supported by Deutsche Telekom, Rita Altura Trust Chair in Computer Sciences, Lynne and William Frankel Center for Computer Sciences, Israel Science Foundation (grant number 428/11), Cabarnit Cyber Security MAGNET Consortium, Grant from the Institute for Future Defense Technologies Research named for the Medvedi of the Technion, and Israeli Internet Association. 


\begin{thebibliography}{10}
	
	\bibitem{DBLP:journals/dc/AfekB93}
	Yehuda Afek and Geoffrey~M. Brown.
	\newblock Self-stabilization over unreliable communication media.
	\newblock {\em Distributed Computing}, 7(1):27--34, 1993.
	
	\bibitem{DBLP:conf/podc/AfekGR92}
	Yehuda Afek, Eli Gafni, and Adi Ros{\'e}n.
	\newblock The slide mechanism with applications in dynamic networks (extended
	abstract).
	\newblock In {\em Proceedings of the Eleventh Annual ACM Symposium on
		Principles of Distributed Computing, Vancouver, British Columbia, Canada,
		August 10-12, 1992}, pages 35--46, 1992.
	
	\bibitem{DBLP:conf/focs/AwerbuchPV91}
	Baruch Awerbuch, Boaz Patt-Shamir, and George Varghese.
	\newblock Self-stabilization by local checking and correction.
	\newblock In {\em 32nd Annual Symposium on Foundations of Computer Science, San
		Juan, Puerto Rico, 1-4 October 1991}, pages 268--277, 1991.
	
	\bibitem{DBLP:conf/pdcat/BeinMY09}
	Doina Bein, Toshimitsu Masuzawa, and Yukiko Yamauchi.
	\newblock Reliable communication on emulated channels resilient to transient
	faults.
	\newblock In {\em 2009 International Conference on Parallel and Distributed
		Computing, Applications and Technologies (PDCAT 2009), Higashi Hiroshima,
		Japan, 8-11 December 2009}, pages 366--371, 2009.
	
	\bibitem{DBLP:conf/wss/BuiDPV99}
	Alain Bui, Ajoy~Kumar Datta, Franck Petit, and Vincent Villain.
	\newblock State-optimal snap-stabilizing {PIF} in tree networks.
	\newblock In {\em 1999 ICDCS Workshop on Self-stabilizing Systems, Austin,
		Texas, June 5, 1999, Proceedings}, pages 78--85, 1999.
	
	\bibitem{DBLP:conf/sss/CournierDLPV10}
	Alain Cournier, Swan Dubois, Anissa Lamani, Franck Petit, and Vincent Villain.
	\newblock Snap-stabilizing linear message forwarding.
	\newblock In {\em Stabilization, Safety, and Security of Distributed Systems -
		12th International Symposium, SSS 2010, New York, NY, USA, September 20-22,
		2010. Proceedings}, pages 546--559, 2010.
	
	\bibitem{DBLP:conf/ipps/CournierDV09}
	Alain Cournier, Swan Dubois, and Vincent Villain.
	\newblock A snap-stabilizing point-to-point communication protocol in
	message-switched networks.
	\newblock In {\em 23rd IEEE International Symposium on Parallel and Distributed
		(IPDPS'09)}, pages 1--11, 2009.
	
	\bibitem{DBLP:journals/jpdc/DelaetDNT10}
	Sylvie Dela{\"e}t, St{\'e}phane Devismes, Mikhail Nesterenko, and S{\'e}bastien
	Tixeuil.
	\newblock Snap-stabilization in message-passing systems.
	\newblock {\em J. Parallel Distrib. Comput.}, 70(12):1220--1230, 2010.
	
	\bibitem{DBLP:journals/cacm/Dijkstra74}
	Edsger~W. Dijkstra.
	\newblock Self-stabilizing systems in spite of distributed control.
	\newblock {\em Commun. ACM}, 17(11):643--644, 1974.
	
	\bibitem{D2K}
	Shlomi Dolev.
	\newblock {\em {Self-Stabilization}}.
	\newblock MIT Press, 2000.
	
	\bibitem{DBLP:journals/ipl/DolevDPT11}
	Shlomi Dolev, Swan Dubois, Maria Potop-Butucaru, and S{\'e}bastien Tixeuil.
	\newblock Stabilizing data-link over non-{FIFO} channels with optimal
	fault-resilience.
	\newblock {\em Inf. Process. Lett.}, 111(18):912--920, 2011.
	
	\bibitem{DBLP:conf/sss/DolevHSS12}
	Shlomi Dolev, Ariel Hanemann, Elad~Michael Schiller, and Shantanu Sharma.
	\newblock Self-stabilizing end-to-end communication in (bounded capacity,
	omitting, duplicating and non-{FIFO}) dynamic networks - (extended abstract).
	\newblock In Andr{\'{e}}a~W. Richa and Christian Scheideler, editors, {\em
		Stabilization, Safety, and Security of Distributed Systems - 14th
		International Symposium, {SSS} 2012, Toronto, Canada, October 1-4, 2012.
		Proceedings}, volume 7596 of {\em Lecture Notes in Computer Science}, pages
	133--147. Springer, 2012.
	
	\bibitem{DBLP:journals/siamcomp/DolevIM97}
	Shlomi Dolev, Amos Israeli, and Shlomo Moran.
	\newblock Resource bounds for self-stabilizing message-driven protocols.
	\newblock {\em SIAM J. Comput.}, 26(1):273--290, 1997.
	
	\bibitem{DBLP:journals/tc/DolevW97}
	Shlomi Dolev and Jennifer~L. Welch.
	\newblock Crash resilient communication in dynamic networks.
	\newblock {\em IEEE Trans. Computers}, 46(1):14--26, 1997.
	
	\bibitem{DBLP:conf/ispan/FlauzacV97}
	Olivier Flauzac and Vincent Villain.
	\newblock An implementable dynamic automatic self-stabilizing protocol.
	\newblock In {\em International Symposium on Parallel Architectures, Algorithms
		and Networks (ISPAN '97), 18-20 December 1997, Taipei, Taiwan}, pages 91--97,
	1997.
	
	\bibitem{forouzan2006data}
	A~Behrouz Forouzan.
	\newblock {\em Data Communications and Networking}.
	\newblock Tata McGraw-Hill Education, 4nd edition, 2006.
	
	\bibitem{DBLP:journals/tc/GoudaM91}
	Mohamed~G. Gouda and Nicholas~J. Multari.
	\newblock Stabilizing communication protocols.
	\newblock {\em IEEE Trans. Computers}, 40(4):448--458, 1991.
	
	\bibitem{DBLP:conf/stoc/KushilevitzOR95}
	Eyal Kushilevitz, Rafail Ostrovsky, and Adi Ros{\'e}n.
	\newblock Log-space polynomial end-to-end communication.
	\newblock In {\em Proceedings of the Twenty-Seventh Annual ACM Symposium on
		Theory of Computing, 29 May-1 June 1995, Las Vegas, Nevada, USA}, pages
	559--568, 1995.
	
	\bibitem{DBLP:journals/cacm/Lamport78}
	Leslie Lamport.
	\newblock Time, clocks, and the ordering of events in a distributed system.
	\newblock {\em Commun. {ACM}}, 21(7):558--565, 1978.
	
	\bibitem{Moon}
	Todd~K. Moon.
	\newblock {\em Error Correction Coding: Mathematical Methods and Algorithms}.
	\newblock Wiley-Interscience, 2005.
	
	\bibitem{reed1960polynomial}
	Irving~S Reed and Gustave Solomon.
	\newblock Polynomial codes over certain finite fields.
	\newblock {\em Journal of the society for industrial and applied mathematics},
	8(2):300--304, 1960.
	
	\bibitem{DBLP:journals/ton/Spinelli97}
	John Spinelli.
	\newblock Self-stabilizing sliding window \textmd{ARQ} protocols.
	\newblock {\em IEEE/ACM Trans. Netw.}, 5(2):245--254, 1997.
	
	\bibitem{stevens2004unix}
	W~Richard Stevens, Bill Fenner, and Andrew~M Rudoff.
	\newblock {\em UNIX network programming}, volume~1.
	\newblock Addison-Wesley Professional, 2004.
	
	\bibitem{DBLP:books/daglib/0008392}
	Andrew~S. Tanenbaum.
	\newblock {\em Computer networks {(4.} ed.)}.
	\newblock Prentice Hall, 2002.
	
	\bibitem{Tel2001}
	Gerard Tel.
	\newblock {\em Introduction to Distributed Algorithms}.
	\newblock Cambridge University Press, New York, NY, USA, 2nd edition, 2001.
	
	\bibitem{Varghese93}
	George Varghese.
	\newblock {\em Self-Stabilization by Local Checking and Correction}.
	\newblock PhD thesis, Laboratory for Computer Science (LCS), Massachusetts
	Institute of Technology (MIT), 1992.
	
\end{thebibliography}

\end{document}